%% file: complexity.tex
\documentclass[10pt]{amsart}
\usepackage[T1]{fontenc}
\usepackage[foot]{amsaddr}

\usepackage{amsmath,amsfonts,amssymb,amsthm,tikz-cd,mathrsfs,mathtools,ragged2e,ifthen}

\usepackage[shortlabels]{enumitem}
\usepackage{hyperref}
\usepackage{todonotes}

\usepackage[capitalise]{cleveref}

\usetikzlibrary{arrows}
\tikzset{commutative diagrams/.cd,arrow style=tikz,diagrams={>=latex'}}

\usepackage{algorithm}
\usepackage[noend]{algpseudocode}
\algnewcommand{\algorithmicassumption}{\textbf{Requirement:}}
\algnewcommand{\Assume}{\item[\algorithmicassumption]}
\algrenewcomment[1]{\hfill \(\triangleright\) \emph{\small #1}} 
\algnewcommand{\InlineIf}[2]{
  \algorithmicif\ #1\ \algorithmicthen\ #2}
\algnewcommand{\InlineElse}[1]{
  \algorithmicelse\ #1}
\algnewcommand{\InlineIfElse}[3]{
  \algorithmicif\ #1\ \algorithmicthen\ #2\ \algorithmicelse\ #3}
\algnewcommand{\InlineFor}[2]{\algorithmicfor\ #1\ \algorithmicdo\ #2} 
\algnewcommand{\CommentLine}[1]{\(\triangleright\) \emph{\small #1}}

\algnewcommand{\algorithmicand}{\textbf{and}}
\algnewcommand{\algorithmicor}{\textbf{or}}
\algnewcommand{\FOR}{\algorithmicfor}
\algnewcommand{\OR}{\algorithmicor}
\algnewcommand{\AND}{\algorithmicand}
\algnewcommand{\IF}{\algorithmicif}
\algnewcommand{\THEN}{\algorithmicthen}
\algnewcommand{\ELSE}{\algorithmicelse}

\newcommand{\ZZ}{\mathbb{Z}}
\newcommand{\ZZpos}{\mathbb{Z}_{>0}}

\renewcommand{\AA}{\mathbb{A}}

\newcommand{\Syz}{\operatorname{Syz}}

\newcommand{\LT}{\operatorname{LT}}
\newcommand{\LM}{\operatorname{LM}}
\newcommand{\gin}{\operatorname{gin}}

\newcommand{\HF}{\operatorname{HF}}
\newcommand{\rank}{\operatorname{rk}}
\newcommand{\rows}{\operatorname{rows}}
\newcommand{\rref}{\operatorname{rref}}
\newcommand{\trace}{\operatorname{trace}}
\newcommand{\im}{\operatorname{im}}

\newcommand{\fraka}{\mathfrak{a}}

\newcommand{\field}{\Bbbk}
\newcommand{\algclosure}{{\overline{\Bbbk}}}
\newcommand{\pring}{\field[x_1,\dots,x_k]}
\newcommand{\ring}{\mathcal{R}}
\newcommand{\mat}{M}
\newcommand{\matmod}{M_n(\ring)}
\newcommand{\idmat}{I_n}
\newcommand{\detidealgen}[1]{\mathcal{I}_{r+1}\left(#1\right)}
\newcommand{\detideal}{\mathcal{I}_{r+1}(\mat)}
\newcommand{\detidealCorkOne}[1][]{\ifthenelse{\equal{#1}{}}{\mathcal{I}_{n-2}(\mat)}{\mathcal{I}_{n-2}(#1)}}
\newcommand{\detsystemCorkOne}[1][]{\ifthenelse{\equal{#1}{}}{F_{n-2}(\mat)}{F_{n-2}(#1)}}
\newcommand{\Mon}{\operatorname{Mon}}
\newcommand{\macmat}{\mathcal{M}}
\newcommand{\macmatred}{\widetilde{\macmat}}
\newcommand{\degdmon}{\Mon_d(\ring)}
\newcommand{\genericring}[1]{\ring\left[\left\{{#1}^{(\tau)}:\tau\in\degdmon\right\}\right]}
\newcommand{\genmat}[1][]{\ifthenelse{\equal{#1}{}}{\mathscr{A}_n^d}{\mathscr{A}_n^{#1}}}
\newcommand{\affspace}[1]{\AA^{#1}}
\newcommand{\affspaceCorkOne}{\affspace{4n^2}}
\newcommand{\degdaffspace}{\affspace{\binom{k+d-1}{k-1}\cdot n^2}}
\newcommand{\property}{\mathscr{P}}
\newcommand{\CM}{\operatorname{CM}}
\newcommand{\RL}{\operatorname{RL}}
\newcommand{\HS}{\operatorname{HS}} 
\newcommand{\modord}{\succ_{\mathrm{TOP}}} 
\newcommand{\LTmod}{\LT_{\modord}} 
\newcommand{\LMmod}{\LM_{\modord}} 
\newcommand{\signature}{\operatorname{sgn}} 
\newcommand{\modf}{\mathbf{f}} 
\newcommand{\modg}{\mathbf{g}} 
\newcommand{\modh}{\mathbf{h}} 
\newcommand{\modF}{\mathbf{F}} 
\newcommand{\modgb}{\textsc{ModGB}} 
\newcommand{\detgb}{\textsc{DetGB}} 
\newcommand{\freeresmod}{\mathcal{E}} 
\newcommand{\cofac}{M^{C}} 

\theoremstyle{definition}
\newtheorem{theorem}{Theorem}[section]
\newtheorem*{theorem*}{Theorem}
\newtheorem{definition}[theorem]{Definition}
\newtheorem{lemma}[theorem]{Lemma}
\newtheorem{proposition}[theorem]{Proposition}
\newtheorem{corollary}[theorem]{Corollary}
\newtheorem{remark}[theorem]{Remark}

\newtheorem{conjecture}[theorem]{Conjecture}

\newcommand{\myparagraph}[1]{\smallskip\emph{#1.}} 
\newlength\ubwidth
\newcommand\parunderbrace[2]{\settowidth\ubwidth{$\displaystyle{#1}$}\underbrace{#1}_{\parbox{1.25\ubwidth}{\scriptsize\RaggedRight#2}}}

\title[Gr\"obner bases of comaximal determinantal ideals]{On the arithmetic complexity of computing Gr\"obner bases of comaximal determinantal ideals}

\author{Sriram Gopalakrishnan}
\email{sriram.gopalakrishnan@lip6.fr}
\address{Author's affiliations: $\begin{cases}\text{Sorbonne Université, CNRS, LIP6, F-75005 Paris, France}\\ \text{University of Waterloo, Waterloo, ON, Canada}\end{cases}$}

\begin{document}

\begin{abstract}
	Let $M$ be an $n\times n$ matrix of homogeneous linear forms over a field
	$\Bbbk$. If the ideal $\mathcal{I}_{n-2}(M)$ generated by minors of size $n-1$
	is Cohen-Macaulay, then the Gulliksen-Neg{\aa}rd complex is a free resolution
	of $\mathcal{I}_{n-2}(M)$. It has recently been shown that by taking into
	account the syzygy modules for $\mathcal{I}_{n-2}(M)$ which can be obtained
	from this complex, one can derive a refined signature-based Gr\"obner basis
	algorithm \textsc{DetGB} which avoids reductions to zero when computing a
	grevlex Gr\"obner basis for $\mathcal{I}_{n-2}(M)$. In this paper, we establish
	sharp complexity bounds on \textsc{DetGB}.  To accomplish this, we prove
	several results on the sizes of reduced grevlex Gr\"obner bases of reverse
	lexicographic ideals, thanks to which we obtain two main complexity results
	which rely on conjectures similar to that of Fr\"oberg. The first one states
	that, in the zero-dimensional case, the size of the reduced grevlex Gr\"obner
	basis of $\mathcal{I}_{n-2}(M)$ is bounded from below by $n^{6}$
	asymptotically. The second, also in the zero-dimensional case, states that the
	complexity of \textsc{DetGB} is bounded from above by $n^{2\omega+3}$
	asymptotically, where $2\le\omega\le 3$ is any complexity exponent for matrix
	multiplication over $\Bbbk$.
\end{abstract}

\maketitle

\input{introduction.tex}
\input{preliminaries.tex}
\input{genericity.tex}
\input{hilb.tex}
\input{staircase.tex}
\input{lefschetz.tex}

\input{algorithm.tex}
\input{analysis.tex}
\input{acknowledgements.tex}

\bibliographystyle{alpha}
\bibliography{complexity}

\end{document}

%% file: introduction.tex
\section{Introduction}

\myparagraph{The MinRank problem} Let $\field$ be a field and let $\algclosure$
be an algebraic closure of \(\field\). Let $\ring=\field[x_1,\dots,x_k]$ for
some $k\in\ZZpos$. Let $M$ be an $m\times n$ matrix whose entries are
homogeneous polynomials in \(\ring\) of degree $d$. Without loss of generality,
suppose $m\ge n$. Let $r\in\ZZpos$ with $r<n$. We denote by $\detideal\subset
\ring$ the ideal generated by the collection of all minors of size $(r+1)$ of
$M$, that is, by all determinants of submatrices of $M$ of size
$(r+1)\times(r+1)$. Note that for any point $x\in V_\algclosure(\detideal)$,
the evaluation of each entry of $M$ at $x$ yields a matrix $M(x)$ with entries
in $\algclosure$ whose rank is at most $r$. Ideals of the form $\detideal$ are
called \emph{determinantal ideals} and have been well-studied (see e.g.\
\cite{BrunsVetter1988}). For $d=1$, meaning that the entries of $M$ are linear
forms, the problem of computing $V(\detideal)$ is known as the \emph{MinRank}
problem. For $d\ge 1$, the problem is known as the \emph{generalized MinRank}
problem. The MinRank problem is known to be $\mathcal{NP}$-hard (see\
\cite{BussFrandsenShallit1999}).

The MinRank problem lies at the heart of many cryptographic schemes (e.g.\
\cite{Courtois2001,Patarin1996,KipnisShamir1999}) and in many cases, it is
possible to reduce the problem of breaking a cryptographic scheme to specific
structured instances of the MinRank problem (see e.g.\
\cite{FaugereLevyPerret2008,DingSchmidt2005,Beullens2022,
BaenaBriaudCabarcasPerlnerSmithToneVerbel2022,
BardetBriaudBrosGaboritNeigerRuattaTillich2020,
BardetBrosCabarcasGaboritPerlnerSmithToneTillitchVerbel2020}).

Outside of the realm of cryptography, many problems in effective algebraic
geometry can be modeled as (generalized) MinRank instances. Problems such as
those of computing critical points (see e.g.\ \cite{FSS12, Spa14}), polynomial
optimization (see e.g.\ \cite{GS14,BGHM14}), quantifier elimination (see e.g.\
\cite{HoSa09,HoSa12,LeSa21}), and others in real algebraic geometry (see e.g.\
\cite{SaSc03, BGHP05, BGHSS, SaSc17,BaSa15, LaSa21}) can all be viewed as
instances of the (generalized) MinRank problem.

\myparagraph{Gr\"obner basis algorithms} One possible technique to solve the
MinRank problem is to solve the polynomial system of $(r+1)$-minors of $M$.
Such systems, known as \emph{determinantal systems}, are well-studied and
highly structured. We refer the reader to
\cite{BrunsVetter1988,Lascoux1978,BCRV22} for a wealth of general theory about
determinantal systems. This structure suggests that existing general polynomial
systems solving techniques could be optimized in search of more efficient
algorithms to solve the MinRank problem. For example in \cite{HSSV21}, specific
properties of determinantal systems are taken into account to present adapted
homotopy continuation techniques for solving determinantal systems.

In this paper, we investigate Gr\"obner basis techniques for solving
determinantal polynomial systems. Many improvements have been made to the
original general-purpose Gr\"obner basis algorithm given by Buchberger in his
thesis \cite{bGroebner1965}. Perhaps one of the most important has been the
introduction of linear algebra via Macaulay matrices to perform $S$-pair
reduction in the $F_4$ algorithm given by Faug\`ere in \cite{Faugere1999}. The
other critical improvement has been the identification and elimination of
reductions to zero by way of a family of criteria, culminating in the $F_5$
algorithm given by Faug\`ere in \cite{Faugere2002}. The $F_5$ algorithm uses at
its core a simple data structure called \emph{signatures}, which keeps track of
the way the Gr\"obner basis was computed. Since the introduction of the $F_5$
algorithm, many signature-based Gr\"obner basis algorithms have been developed.
We refer to \cite{EderFaugere2016} for a survey of such algorithms.

\myparagraph{Arithmetic complexity} We work with the arithmetic complexity
model, counting basic arithmetic operations in $\field$ to estimate the
complexity of Gr\"obner basis algorithms. Under this model of computation, the
complexity of computing Gr\"obner bases using a linear-algebra based algorithm
boils down to that of echelonizing Macaulay matrices over $\field$. One key
requirement in order to obtain reasonable complexity estimates is sharp upper
bounds on the degrees of the polynomials comprising the output Gr\"obner basis.
For (non-determinantal) polynomial systems satisfying certain genericity
properties, such bounds have been given in \cite{Lazard1983,Giusti1984}.

For determinantal ideals, such bounds have been given in
\cite{FSS10,FaugereSafeySpaenlehauer2013}, together with corresponding
complexity analyses which simply compute the cost of echelonizing the
appropriate Macaulay matrices. In these works, the impact of the specific
structure of determinantal ideals on the shapes of the Macaulay matrices is not
taken into account. In particular, reductions to zero arising from the
structure of determinantal ideals are not studied or exploited. These
reductions to zero are in direct correspondence with elements in the syzygy
module of the considered determinantal ideal.  In \cite{GoNeSa23}, these syzygy
modules are described and used to introduce new criteria which avoid reductions
to zero when computing Gr\"obner bases of determinantal ideals. In particular,
when $r=n-2$, $m=n$, and under suitable genericity assumptions, all reductions
to zero are avoided. This implies that all linearly redundant rows in the
Macaulay matrices can be pruned a priori, leading to faster echelonization. A
subsequent complexity analysis is given, which takes into account these
reductions to zero, but still does not exploit the specific structure of the
Macaulay matrices encountered.

\myparagraph{Main results} In this paper, we consider the \emph{comaximal case}
for square matrices of linear forms. That is, we take $d=1$, $r=n-2$, and
$m=n$. The codimension of $\detideal$ is $k-(n-r)^2$ (see e.g.\
\cite[Theorem~10]{FaugereSafeySpaenlehauer2013}). We fix therefore $k=4$, so
that $\detidealCorkOne$ is of dimension zero. In this setting, under certain
genericity assumptions which we make explicit, and assuming certain conjectures
related to the generic grevlex staircase of $\detidealCorkOne$, we give two
main results.

First, we provide an exact formula for the size of the reduced grevlex
Gr\"obner basis of $\detidealCorkOne$ (see \cref{thm:analysis:lower-bound} for
a precise statement and proof): 

\begin{quotation}
	Under certain genericity assumptions and assuming that the ideal
	$\detidealCorkOne$ is reverse lexicographic, the number of elements of
	$\field$ in the dense representation for the reduced grevlex Gr\"obner
	basis of $\detidealCorkOne$ is asymptotically bounded from below by
	$n^{6}$.
\end{quotation}

Second, we give a sharp complexity analysis of the algorithm
\cite[Algorithm~3]{GoNeSa23}-which we call \textsc{DetGB}-taking into account
the specific structure of the Macaulay matrices encountered to obtain our
complexity bound (see \cref{thm:analysis:final-complexity} for the precise
statement and proof): 

\begin{quotation}
	Under certain genericity assumptions and assuming that the ideal
	$\detidealCorkOne$ is reverse lexicographic, the number of arithmetic
	operations in $\field$ performed by \cite[Algorithm~3]{GoNeSa23} when
	computing the reduced grevlex Gr\"obner basis for $\detidealCorkOne$ is in
	$O(n^{2\omega+3})$.
\end{quotation}

To accomplish both of these analyses, we establish results on the structure of
the grevlex staircase of a well-studied class of ideals known as reverse
lexicographic ideals. In \cref{sec:lefschetz}, we rely on the well-studied
notion of \emph{Lefschetz properties} (see \cite{Harima2013ch3}) to relate
\cref{conj:revlex:det-ideal-is-revlex} which states that $\detidealCorkOne$ is
generically reverse lexicographic to the $2$-SLP.

The statement and proof of \cref{thm:analysis:lower-bound} is obtained by
simply determining the size of the output reduced grevlex Gr\"obner basis. In
particular, we use a dense representation of the polynomials in the Gr\"obner
basis and provide a formula for the number of nonzero coefficients appearing in
this dense representation.

Inspired by the sharp complexity analysis of $F_5$ in the case of a regular
sequence in simultaneous Noether position given in
\cite{BardetFaugereSalvy2015}, we again use our results on on the structure of
the grevlex staircases of reverse lexicographic ideals to establish our
complexity upper bound, stated in \cref{thm:analysis:final-complexity},
assuming \cref{conj:revlex:det-ideal-is-revlex}. Our upper bound arises by
first giving an explicit estimate for the number of arithmetic operations over
$\field$ performed in \cite[Algorithm~3]{GoNeSa23}, then analyzing the
asymptotics of this formula. We conclude by showing that the asymptotic
analysis we perform is sharp, by comparing it to the explicit estimate we give.
The upper bound we obtain compares favorably to the bound $O(n^{4\omega+2})$ of
\cite[Theorem~20]{FaugereSafeySpaenlehauer2013}.

%% file: preliminaries.tex
\section{Preliminaries}

Throughout, we denote by $\ring$ the ring $\field[x_1,\dots x_k]$, by
$\Mon(\ring)$ the set of all monomials of $\ring$, and by $\degdmon$ the set of
monomials of degree $d$ of $\ring$. We use the standard multi-index notation,
whereby for some $k$-tuple of integers
$\alpha=(\alpha_1,\dots,\alpha_k)\in\ZZ_{\ge 0}^{k}$, we abbreviate
$x^{\alpha}=x_1^{\alpha_1}\cdots x_k^{\alpha_k}$. We use $\affspace{s}$, to
denote the affine space of dimension $s$ over $\algclosure$, viewed as an
affine variety with the Zariski topology.

\subsection{Gr\"obner bases}
A wealth of general theory about Gr\"obner bases can be found, for example, in
\cite{CoxLittleOShea2005} and the references therein. We recall here only what
is necessary for our purposes.

We denote by $\succ$ an admissible monomial order on $\ring$. That is,
$\succ$ is a total order on $\Mon(\ring)$ such that if $\sigma,\tau$ are
monomials with $\sigma\succ\tau$, then for any monomial $m\in\Mon(\ring)$,
$m\sigma\succ m\tau$, and for which there is no infinitely decreasing sequence
of monomials. Given a polynomial $f\in\ring$, we use $\LM_{\succ}(f)$ to denote
the leading monomial of $f$ with respect to $\succ$, and $\LT_{\succ}(f)$ to
denote the leading term of $f$ with respect to $\succ$, that is,
$\LM_{\succ}(f)$ multiplied by its coefficient in $f$. Given a set
$F\subseteq\ring$ of polynomials, we define the sets
\[
	\LM_{\succ}(F):=\{\LM_{\succ}(f):f\in F\}\quad\text{and}\quad\LT_{\succ}(F):=\{\LT_{\succ}(f):f\in F\}
.\] 
When $\succ$ is clear from the context, we remove it as a subscript.

We use $e_i$ to denote the standard $i$-th basis element of the free
$\ring$-module $\ring^{m}$. A monomial of $\ring^{m}$ is an element of
$\ring^{m}$ of the form $x^{\alpha}e_i$, where $x^{\alpha}\in\Mon(\ring)$ and
$e_i$ is some standard basis element of $\ring^{m}$. We use $\Mon(\ring^{m})$
to denote the set of all monomials of $\ring^{m}$.

The free module $\ring^{m}$ carries a standard grading induced by the grading
by degree on $\ring$. That is, we can write
$\ring^{m}=\bigoplus_{d=0}^{\infty}\ring_d^{m}$, where
\[
\ring_d^{m}=\{f_1e_1+\dots+f_me_m:\deg(f_i)=d\text{ for all }1\le i\le m\}
\] 
is the additive group of \emph{homogeneous elements of degree $d$}. Note that
implicitly, we take $0\in\ring^{m}_d$ for all $d\ge 0$. With respect to this
grading, a monomial $x^{\alpha}e_i\in\Mon(\ring^{m})$ has degree
$\deg(x^{\alpha})$. We use $\Mon_d(\ring^{m})$ to denote the set of all
monomials of $\ring^{m}$ of degree $d$.

The monomial order $\succ$ on $\ring$ induces the \emph{term over position}
order on $\ring^{m}$ defined as follows: for $x^{\alpha}e_i,x^{\beta}e_j$
in $\Mon(\ring^{m})$, $x^{\alpha}e_i\modord x^{\beta}e_j$ if and only if either
$x^{\alpha}\succ x^{\beta}$ or $x^{\alpha}=x^{\beta}$ and $i>j$. We likewise
extend the leading term and leading monomial notation from $\ring$ to
$\ring^{m}$, so that for $\modf\in\ring^{m}$, $\LMmod(\modf)$ is the leading
monomial of $\modf$ and $\LTmod(\modf)$ is its leading term. Analogously, for a
subset $\modF\subseteq\ring^{m}$, we denote
$\LMmod(\modF)=\{\LMmod(\modf):\modf\in\modF\}$ and
$\LTmod(\modF)=\{\LTmod(\modf):\modf\in\modF\}$.

Given $\modf_1,\dots,\modf_s\in\ring^{m}$, we denote by
$\langle\modf_1,\dots,\modf_s\rangle$ the submodule of $\ring^{m}$ generated by
$\modf_1,\dots,\modf_s$.

If $\modf_1,\dots,\modf_s\in\ring^{m}$ are all homogeneous elements, then the
module $\modF=\langle\modf_1,\dots,\modf_s\rangle$ is itself graded. In this
setting, we denote by $\modF_d$ the additive group of homogeneous elements of
degree $d$ of $\modF$. 

\begin{definition}[Gr\"obner basis, {\cite[Chapter~5, Definition~2.6]{CoxLittleOShea2005}}]
	Let $\modf_1,\dots,\modf_s$ in $\ring^{m}$ and let $\modF$ be the
	submodule of $\ring^{m}$ generated by $\modf_1,\dots,\modf_s$. A finite
	set $G\subseteq\modF$ is called a \emph{$\modord$-Gr\"obner basis} of
	$\modF$ if $\langle\LMmod(G)\rangle=\LMmod(\modF)$.

	Suppose $\modf_1,\dots,\modf_s$ are homogeneous elements of
	$\ring^{m}$. Then $\modF$ is itself a graded module. In this case, for
	a given integer $D$, a finite set $G\subseteq\modF$ is called a
	\emph{($D,\modord$)-Gr\"obner basis} of $\modF$ if $G$ forms the
	elements of degree at most $D$ of a $\modord$-Gr\"obner basis of
	$\modF$.
\end{definition}
\begin{remark}
	If $G$ is a $\modord$-Gr\"obner basis of a submodule
	$\modF\subseteq\ring^{m}$, then $\langle G\rangle=\modF$ (see
	\cite[Chapter~5, Proposition~2.7(b)]{CoxLittleOShea2005}).
\end{remark}

\subsection{Macaulay matrices}

We recall here the basic construction of Macaulay matrices for a system of
homogeneous elements of $\ring^{m}$.

\begin{definition}
	Let $\mathbf{f} = (\modf_1,\dots,\modf_s)\subseteq\ring^{m}$ be homogeneous
	elements and for each $1\le i\le s$, let $d_i=\deg \modf_i$ . For a given
	integer $d\ge\min_{1\le i\le s}d_i$, the \textit{Macaulay matrix in degree $d$
	with respect to $\modord$}, $\macmat_{d, \succ}(\mathbf{f})$ is constructed
	as follows:
	\begin{itemize}
		\item its columns are indexed by $\Mon_d(\ring^{m})$, ordered in decreasing
		order with respect to $\modord$,
		\item and for each monomial $\tau\in\Mon_{d-d_i}(\ring)$, one inserts into
		the matrix a single row whose entry in the column indexed by the
		monomial $\sigma$ is the coefficient of $\sigma$ in $\tau \modf_i$.
	\end{itemize}
\end{definition}

The rows of $\macmat_{d, \succ}(\mathbf{f})$ are naturally interpreted as
module elements, and we will freely refer to them as such. Note that these rows
form a basis for $\ring^{m}_d$ as a finite-dimensional $\field$-vector space.

We denote by $\macmatred_{d,\succ}(\modf)$ the reduced row-echelon form of
$\macmat_{d,\succ}(\modf)$.

\begin{theorem}\label{thm:preliminaries:macmat-echelonization}
	Let $\modf=(\modf_1,\dots,\modf_s)\subseteq\ring^{m}$ be homogeneous
	elements. Then the rows of $\macmatred_{d, \succ}(\mathbf{f})$, form
	the elements of degree $d$ of a $\modord$-Gr\"obner basis for the
	module generated by $\modf$.
\end{theorem}
\begin{proof}
	Let $\modF$ be the module generated by $\modf$. Let $\modg\in\modF$.
	Since $\modf_1,\dots,\modf_s$ are homogeneous, the module $\modF$ is
	graded. Thus, there exist some $\modg_1,\dots,\modg_t\in\modF$ such
	that for each $1\le i\le t$, $\modg_i$ is homogeneous of degree $i$ and
	$\modg=\modg_1+\dots+\modg_t$.
	For each $1\le i\le t$, the rows of $\macmatred_{i,\succ}(\modf)$ form
	a basis for $\modF_i$. Therefore, for each $1\le i\le t$, there exist
	some $c_j\in\field$ such that 
	\[
		\modg_i=\sum_{j=1}^{\dim_\field(\modF_i)}c_j\modh_j
	\] 
	where the module elements $\modh_j$ are the (nonzero) rows of
	$\macmatred_{i,\succ}(\modf)$. Since $\macmatred_{d,\succ}(\modf)$ is
	in row echelon form, the leading monomials of the $\modh_j$ are
	pairwise distinct. Using a natural generalization of
	\cite[Chapter~2, Lemma~8(ii)]{CoxLittleOShea2007} to the setting of
	modules, there exists, for each $1\le i\le t$, some $1\le
	j\le\dim_\field(\modF_i)$ such that $\LMmod(\modg_i)=\LMmod(\modh_j)$.
	Now since the degrees of the $\modg_i$ are pairwise distinct, using
	once again a natural generalization of \cite[Chapter~2,
	Lemma~8(ii)]{CoxLittleOShea2007}, there exists some $1\le i\le t$ such
	that $\LMmod(\modg)=\LMmod(\modg_i)$.  
\end{proof}

Henceforth, we fix $\succ$ to be the \emph{graded reverse lexicographic}
(grevlex) order. When $\modf$ is clear from context, $\macmat_{d,
\succ}(\mathbf{f})$ and $\macmatred_{d, \succ}(\mathbf{f})$ will be denoted by
$\macmat_d$ and $\macmatred_{d}$ respectively.

\subsection{Hilbert series}

Since the complexity of computing Gr\"obner bases is governed by that of
echelonizing Macaulay matrices, understanding the sizes of these matrices is
key in our complexity analysis. Hilbert series of graded modules encode
precisely this information. Much general theory of Hilbert functions,
polynomials, and series can be found, for example, in \cite[Chapter~6,
Section~4]{CoxLittleOShea2005}.

\begin{definition}
	Given homogeneous elements $\modf_1,\dots,\modf_s\in\ring^{m}$, the
	\emph{Hilbert function} of the module $\modF$ generated by
	$(\modf_1,\dots,\modf_s)$ is defined by
	\[
		\HF_{\modF}(d)=\dim_\field(\modF_d)
	.\] 
	The \emph{Hilbert series} of $\modF$ is the generating series of the
	Hilbert function of $\modF$. That is,
	\[
		H_{\modF}(t)=\sum_{d\ge 0}\HF_{\modF}(d)t^d
	.\] 
\end{definition}

\begin{proposition}[{\cite[Corollary~4.1.8]{BrunsHerzog1998}}]
	Let $\modF\subseteq\ring^{m}$ be a graded module of Krull dimension
	zero. Then $H_{\ring^{m}/\modF}(t)$ is a polynomial.
\end{proposition}

\subsection{Free resolutions}
Free resolutions are a fundamental construction in commutative algebra. Many
general facts about free resolutions can be found in \cite[III]{Eisenbud1995}
and \cite[Chapter~6]{CoxLittleOShea2007}. Again, we recall below only what we
need for our purposes, and follow closely the exposition given in
\cite[Section~2]{gopalakrishnan2024optimized}

Let $\modF$ be a finitely generated $\ring$-module. An exact sequence
\[
	\cdots\xrightarrow{\partial_{j+1}}\freeresmod_j\xrightarrow{\partial_j}\cdots\xrightarrow{\partial_2}\freeresmod_1\xrightarrow{\partial_1}\freeresmod_0\xrightarrow{\epsilon}\modF\to 0
\] 
is a \emph{left resolution} of $\modF$. The maps $\partial_i$ are
\emph{boundary homomorphisms}, and the map $\epsilon$ is an \emph{augmentation
homomorphism}. If for each $i$, the module $\freeresmod_i$ is free, then the
resolution is a \emph{free resolution}. For the sake of brevity, we will often
refer to a resolution as above simply by
$(\freeresmod_\bullet\xrightarrow{\epsilon}\modF, \partial_\bullet)$. We call
$\sup\{i\in\ZZ:\freeresmod_i\ne 0\}$ the \emph{length} of the resolution
$(\freeresmod_\bullet\xrightarrow{\epsilon}\modF, \partial_\bullet)$. Note that
the length of $(\freeresmod_\bullet\xrightarrow{\epsilon}\modF,
\partial_\bullet)$ could be infinity. Free resolutions of finite length are
\emph{finite free resolutions}.

\begin{theorem}[Hilbert's syzygy theorem, {\cite[Corollary~19.7]{Eisenbud1995}}]
	Let $\modF$ be a finitely generated $\ring$-module. There exists a free
	resolution $(\freeresmod_\bullet\xrightarrow{\epsilon}\modF,
	\partial_\bullet)$ of length at most $k$.
\end{theorem}

When $\modF$ is a graded module over $\ring$, it possesses a free resolution
$(\freeresmod_\bullet\xrightarrow{\epsilon}\modF, \partial_\bullet)$ where each
free module $\freeresmod_i$ is graded in such a way that the boundary maps
$\partial_i$ and the augmentation map $\epsilon$ are graded $\ring$-module
homomorphisms. In this setting, those free resolutions
$(\freeresmod_\bullet\xrightarrow{\epsilon}\modF, \partial_\bullet)$ such that
the ranks of each of the $\freeresmod_i$ are minimal are \emph{minimal free
resolutions}.

Let $(\freeresmod_\bullet\xrightarrow{\epsilon}\modF,\partial_\bullet)$ be a
free resolution of $\modF$. Upon fixing a set of generators
$(\modf_1^{(0)},\dots,\modf_{s_0}^{(0)})$ for $\modF$ and sets of generators
$(\modf_1^{(i)},\dots,\modf_{s_i}^{(i)})$ for each $\im\partial_i$, we can
define, for each $i$, the \emph{$i$-th syzygy module} of
$(\modf_1,\dots,\modf_s)$, as follows:
\[
	\Syz_i(\modF)=\{(g_1,\dots,g_{s_i})\in\ring^{s_i}:g_1\modf_1^{(i)}+\dots+g_{s_i}\modf^{(i)}_{s_i}=0\}
.\] 

\begin{remark}
	The way we define syzygy modules here is quite ad-hoc. In particular,
	as defined here, the syzygy modules of a given module $\modF$ depend on
	the choices of a free resolution of $\modF$ and generating sets for the
	images of each of the boundary homomorphisms in the chosen free
	resolution. In our setting this suffices, since we consider only a
	single free resolution\,---\,the Gulliksen-Neg{\aa}rd complex\,---\,and take
	explicit generating sets for the images of the boundary homomorphisms
	from \cite[Theorem~9]{GoNeSa23} and \cite[Proposition~19]{GoNeSa23}.
\end{remark}

The connection between free resolutions and Hilbert series is elucidated in the
following corollary.

\begin{corollary}{{\cite[Theorem~4.4]{CoxLittleOShea2005}}}\label{cor:prelim:hilb-free-res}
	Let $\modF$ be a finitely generated graded $\ring$-module, and let
	$(\freeresmod_\bullet\xrightarrow{\epsilon}\modF, \partial_\bullet)$ be
	a finite graded free resolution of $\modF$ of length $\ell$. For any
	$1\le i\le\ell$, let $s_i=\rank(\freeresmod_i)$ and write
	$\freeresmod_i=\bigoplus_{j=1}^{s_i}\ring\left(-d_i^{(j)}\right)$. Then
	\[
		\HF_\modF(d)=\sum_{i=0}^{\ell}(-1)^{i}\left(\sum_{j=1}^{s_i}\binom{k+d-d_i^{(j)}-1}{k-1}\right)
	.\] 
\end{corollary}

%% file: genericity.tex
\section{Genericity}\label{sec:genericity}

We begin by fixing a notion of genericity for determinantal ideals on which we
will rely to give subsequent results regarding the structure of Gr\"obner bases
of determinantal ideals and the complexity of computing them.

If $M$ is an $n\times n$ matrix over any ring $A$ and $1\le r<n$, we will
denote by $\detideal$ the ideal generated by the $(r+1)$-minors of $M$. We use
the notation of \cite[Section~2]{FaugereSafeySpaenlehauer2013} to formalize
various notions of genericity.

For some integer $d\in\ZZpos$, we call a polynomial of the form
\[
	f=\sum_{\tau\in\degdmon}\fraka^{(\tau)}\tau\in\genericring{\fraka}
\] 
a \emph{generic homogeneous polynomial of degree $d$}.

In what follows, we use $\genmat$ to denote the $n\times n$ matrix of generic
homogeneous polynomials of degree $d$, whose $i,j$ entry is
\[
	 f_{ij}=\sum_{\tau\in\degdmon}\fraka_{ij}^{(\tau)}\tau\in\genericring{\fraka_{ij}}
.\] 
For a point $a=\left(a_{ij}^{(\tau)}\right)\in\degdaffspace$ we will denote by
$\phi_a$ the specialization map
\begin{align*}
	\phi_a:\genericring{\fraka_{ij}}&\to\ring\\
	\fraka_{ij}^{(\tau)}&\mapsto a_{ij}^{(\tau)}
.\end{align*}

By abuse of notation, we will also use $\phi_a\left(\genmat\right)$ to denote the matrix
(over  $\ring$) whose entries are simply the images of the entries of $\genmat$
under $\phi_a$.

A map
\[
	\property:\{\text{Ideals of }\ring\}\to\{\texttt{true},\texttt{false}\}
\] 
is called a \textit{property}.

\begin{definition}
	A property $\property$ is  \textit{$(k,r,n,d)$-generic} if there exists
	a nonempty Zariski open subset $U\subseteq\degdaffspace$ such that for
	all $a\in U$,
	\[
		\property\left(\detidealgen{\phi_a\left(\genmat\right)}\right)=\texttt{true}
	.\] 
\end{definition}

The following proposition is a fundamental result on determinantal ideals.

\begin{proposition}\label{prop:genericity:det-ideal-is-cm}
	Let $\CM$ be the property given by
	 \[
		 \CM(I)=
		 \begin{cases}
			 \texttt{true}\quad &\text{if } I\text{ is Cohen-Macaulay}\\
			 \texttt{false} \quad &\text{otherwise}
		 \end{cases}
	\] 
	Then for any $n,d\in\ZZpos$, $r<n$, $\CM$ is $((n-r)^2,r,n,d)$-generic.
\end{proposition}
\begin{proof}
	By \cite[Theorem~2.5]{BrunsVetter1988}, the depth of
  $\detidealgen{\phi_a(\genmat)}$ is at most $(n-r)^2$. By
  \cite[Theorem~10]{FaugereSafeySpaenlehauer2013}, there exists a Zariski open
  subset $U\subseteq\degdaffspace$ such that for all $a\in U$,
  $\detidealgen{\phi_a\left(\genmat\right)}$ has codimension $(n-r)^2$, and is therefore 
  Cohen-Macaulay.
\end{proof}

%% file: hilb.tex
\section{The Hilbert series of determinantal ideals}

One of the key inputs into our complexity analysis is the Hilbert series of
determinantal ideals. Here, we give explicit formul{\ae} for the Hilbert
functions of determinantal ideals in the case $r=n-2$, under certain genericity
assumptions. We accomplish this by analyzing a complex associated to these
ideals called the Gulliksen-Neg{\aa}rd complex, which under suitable genericity
assumptions turns out to be a free resolution. This allows us to obtain the
Hilbert series' we require by computing the ranks of the component modules.

\subsection{The complex of Gulliksen and Neg{\aa}rd}\label{subsec:gn}

We begin by recalling the construction of the complex of Gulliksen and
Neg{\aa}rd. This complex was originally given in \cite{GulliksenNegard1972}. A
detailed exposition of this complex can be found in \cite[Chapter~2,
Section~D]{BrunsVetter1988}. We reproduce here only what is necessary to obtain
the Hilbert series' we require.

In this section, we fix $n\in\ZZ$, $n\ge 3$. Let $\matmod$ be the set of
$n\times n$ matrices over $\ring$. The set $\matmod$ carries a natural
$\ring$-module structure, under which it is free of rank $n^2$. We will denote
by $\idmat$ the identity matrix in $\matmod$.

Consider the zero sequence
\[
	\ring\xrightarrow{\iota}\matmod\oplus\matmod\xrightarrow{\pi}\ring
\] 
where $\iota(x)=(x\idmat, x\idmat)$ and  $\pi(U,V)=\trace(U-V)$. It is immediate
that $\im\iota\subset\ker\pi$.

\begin{proposition}\label{prop:hilb:gn-E2-rank}
	The quotient module $\mathcal{E}_1=\ker\iota/\im\pi$ is free of rank
	$2n^2-2$.
\end{proposition}
\begin{proof}
	We take as an $\ring$-module basis for $\matmod$ the elementary
	matrices: for each $1\le i,j\le n$, $E_{i,j}$ is the $n\times n$ matrix
	over $\ring$ with entry $1$ at $(i,j)$ and $0$ elsewhere.  

	Following \cite[Chapter~2, Section~D]{BrunsVetter1988}, $\ker\pi$ is
	generated by the following $2n^2-1$ elements of $\matmod\oplus\matmod$
	\begin{itemize}
		\item $(E_{i,j},0)$ for $1\le i,j,\le n$ and $i\ne j$
		\item $(0,E_{u,v})$ for $1\le u,v\le n$ and $u\ne v$
		\item $(E_{i,i},E_{1,1})$ for $1\le i\le n$
		\item $(0,E_{i,i}-E_{1,1})$ for $1\le i\le n$
	\end{itemize}
	It is clear that $\im\iota$ is generated by
	\[
		(\idmat,\idmat)=\sum_{i=1}^{n}(E_{i,i},E_{1,1})+(0,E_{i,i}-E_{1,1})
	.\] 
	Thus, $\ker\iota/\im\pi$ is free of rank $2n^2-2$.
\end{proof}

Equipped with $\mathcal{E}_1$, we can now give the full Gulliksen-Neg{\aa}rd
complex.

Let $M\in\matmod$ and let $M^{C}$ be the matrix of cofactors of $M$. Let
$\mathcal{E}_3=\ring$, $\mathcal{E}_2=\mathcal{E}_0=\matmod$, and
$\freeresmod_1=\ker\psi/\im\phi$. 

We define the boundary maps as follows: $\partial_3:x\mapsto x\cofac$,
$\partial_2:N\mapsto\overline{(MN,NM)}$,
$\partial_1:\overline{(N_1,N_2)}\mapsto N_1M-MN_2$. Finally, we define the
augmentation map by $\epsilon:N\mapsto\trace(M^{C}N)$.
\begin{proposition}[{\cite[Theorem~2.26]{BrunsVetter1988}}]\label{prop:hilb:gn-free-res}
	Let $M\in\matmod$. If $\detidealCorkOne$ is Cohen-Macaulay, then the sequence
	\[
		0\to\freeresmod_3\xrightarrow{\partial_3}\freeresmod_2\xrightarrow{\partial_2}\freeresmod_1\xrightarrow{\partial_1}\freeresmod_0\xrightarrow{\epsilon}\detidealCorkOne\to 0
	\]
	with modules, boundary maps, and augmentation map defined as above is a
	free resolution of $\detidealCorkOne$.
\end{proposition}

\subsection{An explicit Hilbert series}

Applying \cref{cor:prelim:hilb-free-res}, we extract from the
Gulliksen-Neg{\aa}rd complex the Hilbert series of $\detidealCorkOne$.

\begin{proposition}\label{prop:hilb:hilb}
	For any $n\ge 3$, $D\ge 1$, let 
	\begin{align*}
		H_{n,D}(t)=\sum_{d=D(n-1)}^{2Dn-3}\Bigg(n^2\binom{3+d-D(n-1)}{3}&-(2n^2-2)\binom{3+d-Dn}{3}\\
										&+n^2\binom{3+d-D(n+1)}{3}\Bigg)t^d
	.\end{align*}
	Then the property
	\[
		\HS(I)=
		\begin{cases}
			\texttt{true}\quad&\text{if }H_I(t)=H_{n,D}(t)\\
			\texttt{false}\quad&\text{otherwise}
		\end{cases}
	\] 
	is $(4,n-2,n,D)$-generic.
\end{proposition}
\begin{proof}
	By \cref{prop:genericity:det-ideal-is-cm}, there exists a Zariski open
	subset $U_\mathrm{CM}\subseteq\affspaceCorkOne$ such that for all $a\in
	U_\mathrm{CM}$, $\detidealCorkOne[\phi_a({\genmat[D]})]$ is
	Cohen-Macaulay. Subsequently, by \cref{prop:hilb:gn-free-res}, for any
	$a\in U_\mathrm{CM}$, the Gulliksen-Neg{\aa}rd complex is a free
	resolution of $\detidealCorkOne[\phi_a({\genmat[D]})]$. In order to
	apply \cref{cor:prelim:hilb-free-res} on the Gulliksen-Neg{\aa}rd
	complex, we must write, for each $0\le i\le 3$,
	$\freeresmod_i=\bigoplus_{j=1}^{\rank(\freeresmod_i)}\ring(-d_i^{(j)})$.

	First, $\rank(\freeresmod_0)=n^2$ and for each $1\le j\le n^2$, the
	image of $e_j$ under $\epsilon$ is an $(n-1)$-minor of
	$\phi_a({\genmat[D]})$. Such a minor is a polynomial of degree $D(n-1)$,
	so in order for $\epsilon$ to be graded, $e_i$ must be a member of the
	$D(n-1)$ graded piece of $\freeresmod_0$. This gives $d_0^{(j)}=D(n-1)$
	for all $1\le j\le n^2$.

	Next, by \cref{prop:hilb:gn-E2-rank}, $\rank(\freeresmod_1)=2n^2-2$.
	Consider the basis elements of $\freeresmod_1$ from
	\cref{prop:hilb:gn-E2-rank} of the form $\overline{(E_{i,j},0)}$ for
	$1\le i,j\le n$, $i\ne j$. These map, under $\partial_1$, to
	$E_{i,j}\phi_a({\genmat[D]})$. Each nonzero coefficient of
	$E_{i,j}\phi_a({\genmat[D]})$, written in terms of the standard basis
	elements of $\freeresmod_0=\matmod$, is an entry of
	$\phi_a({\genmat[D]})$, and is thus a polynomial of degree $D$. Thus, in
	order for the map $\partial_1:\freeresmod_1\to\freeresmod_0(-D(n-1))$ to
	be graded, $\overline{(E_{i,j},0)}$ must be a member of the $Dn$ graded
	piece of $\freeresmod_1$. This gives $d_1^{(1)}=Dn$. A similar
	computation for the other basis elements of $\freeresmod_1$, described
	in \cref{prop:hilb:gn-E2-rank} shows that, in fact, for all $1\le j\le
	2n^2-2$, $d_1^{(j)}=Dn$.

	By construction, $\rank(\freeresmod_2)=n^2$. For any basis element
	$E_{i,j}$ of $\freeresmod_2$, each nonzero coefficient of
	$\partial_2(E_{i,j})=\overline{(\phi_a({\genmat[D]})E_{i,j},E_{i,j}\phi_a({\genmat[D]}))}$,
	written in terms of the basis elements for $\freeresmod_1$ given in
	\cref{prop:hilb:gn-E2-rank} is again an entry of $\phi_a({\genmat[D]})$,
	and therefore a polynomial of degree $D$. In order for
	$\partial_2:\freeresmod_2\to\freeresmod_1(-Dn)$ to be graded, $E_{i,j}$
	must then be a member of the $D(n+1)$ graded piece of $\freeresmod_2$.
	This gives, for each $1\le j\le n^2$, $d_2^{(j)}=D(n+1)$.

	Finally, $\rank(\freeresmod_3)=1$, and
	$\partial_3(1)=\phi_a({\genmat[D]})^{C}$. When written in terms of the
	standard basis of $\freeresmod_2=\matmod$, the nonzero coefficients of
	$\phi_a({\genmat[D]})^{C}$ are precisely the cofactors of
	$\phi_a({\genmat[D]})$, which have degree $D(n-1)$, giving
	$d_3^{(1)}=2Dn$.

	Putting these quantities together and applying
	\cref{cor:prelim:hilb-free-res} gives precisely $H_{n,D}(t)$.
\end{proof}

We conclude this section with an auxiliary lemma regarding the structure of
the Hilbert series given in \cref{prop:hilb:hilb} which will be necessary in
order to perform the complexity analyses we give in later sections.

\begin{lemma}\label{lem:analysis:h-increasing}
	For all $n-1\le d<2n-3$, $h_{d+1}>h_d$. 
\end{lemma}
\begin{proof}
	Taking $D=1$ in \cref{prop:hilb:hilb} and simplifying the summand, we
	find that
	\[
		h_d=\frac{(2+d-n)(d^2+(-2n+4)d+4n^2-4n+3)}{3}
	.\] 
	Let $f(d)=h_d$ and view $f$ as a continuous function in one real
	variable. Then $f''(d)=2d-2n+4$ has its unique root at $d=n-2$. Thus,
	$f'(d)$ attains its minimum at $d=n-2$, and $f$ is therefore strictly
	increasing on the interval $[n-1,\infty)$.
\end{proof}

%% file: staircase.tex
\section{Reverse lexicographic ideals}
\label{sec:revlex_ideals}

\emph{Reverse lexicographic ideals} have been studied for the special place
they hold amongst all ideals with a given Hilbert function (see e.g.\
\cite{Deery96}). After recalling their definition in
\cref{sec:revlex_ideals:grevlex_staircase}, we prove a general result about the
structure of the grevlex leading monomials in the case of a homogeneous reverse
lexicographic ideal (\cref{prop:revlex:new-gb-bound}). In
\cref{sec:revlex_ideals:determinantal}, we show that as long as a certain
Zariski open subset that we make explicit is nonempty, this result can be
applied to the determinantal ideals considered in this paper. In the following
section, \cref{sec:lefschetz}, we further investigate when the determinantal
ideals that we consider are generically reverse lexicographic, and relate this
property to the so-called $2$-\emph{strong Lefschetz property}.

Recall that we have fixed our monomial order $\succ$ to be the graded reverse
lexicographic order and suppressed the symbol $\succ$ from all relevant
notation.

\subsection{The grevlex staircase of a homogeneous reverse lexicographic ideal}
\label{sec:revlex_ideals:grevlex_staircase}

\begin{definition}
	Let $I\subseteq\ring$ be a nonzero ideal. We call $I$ \emph{reverse
	lexicographic} if for all $\tau\in\LM(I)$, 
	\[
		\{\sigma\in\Mon_{\deg\tau}(\ring):\sigma\succ\tau\}\subseteq\LM(I)
	.\] 
\end{definition}

We begin with a simple and helpful observation about the grevlex staircase of a
reverse lexicographic ideal.

\begin{lemma}\label{lem:revlex:collisions}
  Let $I\subseteq\ring$ be a reverse lexicographic ideal. Let $\tau\in\LM(I)$
  and $x_j$ be the smallest variable in $\tau$. For any variable \(y \in
  \{x_1,\ldots,x_k\}\) with $y\succ x_j$, the monomial \(\sigma =
  \frac{\tau}{x_j} y \in \Mon_{\deg\tau}(\ring)\) is such that
  $\sigma\succ\tau$ and $\sigma\in\LM(I)$.
\end{lemma}

We come now to the main result of this section, which provides a formula to
calculate the number of polynomials of degree $d+1$ in a grevlex Gr\"obner
basis of a homogeneous reverse lexicographic ideal whose leading monomials are
not divisible by any leading monomial of degree $d$.
\begin{proposition}\label{prop:revlex:new-gb-bound}
	Let $F\subseteq\pring$ be a sequence of homogeneous polynomials, all of
	degree $d_0$. Suppose $I=\langle F\rangle$ is a reverse lexicographic
	ideal. Let $\HF_{I}(d)=h_d$ be the Hilbert function of $I$ and let $D$
	be the largest degree of a polynomial appearing in the reduced grevlex
	Gr\"obner basis of $F$. For any $d_0\le d\le D$, let $G_d$ be the set of
	elements of degree at most $d$ in the reduced grevlex Gr\"obner basis for
	$F$ and let $\ell_d$ be the largest integer such that
	\[
		\binom{\ell_d+d-1}{\ell_d-1}<h_d
	.\]
	Then for any $d_0\le d\le D$,
	\begin{align*}
		\#(\LM(G_{d+1})\smallsetminus\LM(G_{d}))=h_{d+1}&+(\ell_d-k)h_d\\
								&+\sum_{j=1}^{\ell_d}\binom{j+d-2}{j-1}(j-1)\\
								&-\ell_d\binom{\ell_d+d-1}{\ell_d-1}
	\end{align*}
\end{proposition}

\begin{proof}
	We begin by noting that $G_{d+1}$ can be obtained from $G_d$ by
	multiplying each of the $h_d$ nonzero rows of $\macmatred_d$ by each of
	the $k$ variables to build $\macmat_{d+1}$, echelonizing to obtain
  $\macmatred_{d+1}$, then discarding zero rows and rows that are redundant
  because their leading terms are already divisible by those in $\LM(G_d)$.
  Letting $z_{d+1}$ (resp.\ $r_d$) be the number of these zero (resp.\
  redundant) rows, we can write \begin{equation}
		\#(\LM(G_{d+1})\smallsetminus\LM(G_d))=kh_d-z_{d+1}-r_{d+1}\label{eq:prop:revlex:countlt:n}
	.\end{equation} 
	Note also that if the echelonization process alters the leading term of
	some row of $\macmat_{d+1}$ built in this way, then either that row
	reduces to zero, or its new leading term is no longer divisible by any
  monomial in $\LM(G_d)$. Thus, denoting by $c_{d+1}$ the number of rows of
  $\macmat_{d+1}$ which are to be reduced during the echelonization process,
  one has
	\begin{gather}
		c_{d+1}=z_{d+1}+\#(\LM(G_{d+1})\smallsetminus\LM(G_d))\label{eq:prop:revlex:countlt:c}\\
		\text{and}\;\;\; r_{d+1}=kh_d-c_{d+1} \label{eq:prop:revlex:countlt:r}.
	\end{gather} 
	Combining
	\cref{eq:prop:revlex:countlt:n,eq:prop:revlex:countlt:c,eq:prop:revlex:countlt:r}
	gives 
	\[
		\#(\LM(G_{d+1})\smallsetminus\LM(G_d))=h_{d+1}-kh_d+c_{d+1}
	.\] 
	The rest of the proof consists in computing $c_{d+1}$. Fix
	$\tau\in\LM(G_d)$, and let $x_j$ be the grevlex smallest variable in
	$\tau$. Then by \cref{lem:revlex:collisions}, for each of the $j-1$
	variables larger than $x_j$, there exists some $\sigma\in\LM(G_d)$ such
	that $\sigma x_j$ appears as the leading term of some row of
	$\macmat_{d+1}$ which can be reduced by a multiple of $\tau$. Thus, the
	row of $\macmatred_d$ with leading term $\tau$ generates exactly $j-1$
	rows of $\macmat_{d+1}$ which are to be reduced. The number of monomials
	of degree $d$ whose grevlex smallest variable is some $x_j$ is simply
	\[
		\parunderbrace{\binom{j+d-1}{j-1}}{number of monomials of degree $d$ in $j$ variables}
    \quad - \quad
    \parunderbrace{\binom{j+d-2}{j-2}}{number of monomials of degree $d$ in  $j-1$ variables}
    \quad = \quad
    \parunderbrace{\binom{j+d-2}{j-1}}{number of monomials of degree $d$ with grevlex smallest variable $x_j$}
	\] 
	Finally, since the leading monomials of $\macmatred_d$ are simply the
	$h_d$ grevlex largest monomials of degree $d$, we obtain the final
	result
	\begin{align*}
		\#(\LM(G_{d+1})\smallsetminus\LM(G_d))&=h_{d+1}-kh_d+c_{d+1}\\
						      &=h_{d+1}-kh_d+\sum_{j=1}^{\ell_d}\binom{j+d-2}{j-1}(j-1)\\
						      &\qquad+\ell_d\left(h_d-\binom{\ell_d+d-1}{\ell_d-1}\right)\\
						      &=h_{d+1}+(\ell_d-k)h_d+\sum_{j=1}^{\ell_d}\binom{j+d-2}{j-1}(j-1)\\
						      &\qquad-\ell_d\binom{\ell_d+d-1}{\ell_d-1}
	.
  \qedhere
  \end{align*} 
\end{proof}

\subsection{Reverse lexicographic determinantal ideals}
\label{sec:revlex_ideals:determinantal}

With the hope of applying \cref{prop:revlex:new-gb-bound}, we investigate here
the conditions under which the determinantal ideals we consider are indeed
reverse lexicographic. We begin by showing that the Macaulay matrices of
reverse lexicographic ideals possess a certain structure. Next, we construct an
explicit Zariski open subset whose points correspond to reverse lexicographic
ideals of the form $\detidealCorkOne$. That this Zariski open subset is
nonempty is left as a conjecture, which we give insight into in
\cref{sec:lefschetz}.

\begin{lemma}\label{lem:revlex:revlex-macaulay}
	Let $F\subseteq\ring$ be a sequence of homogeneous polynomials. Then
	$I=\langle F\rangle$ is a reverse lexicographic ideal if and only if
	for any $d\in\ZZpos$, the first $h_d$ columns of the Macaulay matrix in
	degree $d$, $\macmat_d$, have rank $h_d=\rank(\macmat_d)$, or
	equivalently, the echelonized Macaulay matrix in degree $d$,
	$\macmatred_d$ takes the form
	\[
	\macmatred_d=
	\left(
		\begin{array}{c|c}
			I & X
		\end{array}
	\right)
	\] 
	after having removed reductions to zero and up to a permutation of
	rows, where $I$ is the identity matrix of size $h_d\times h_d$.
\end{lemma}
\begin{proof}
	Fix $d\in\ZZpos$.
	Suppose $I$ is reverse lexicographic, and let
	\[
	\tau_0=\min\{\tau\in\LM(I):\deg(\tau)=d\}
	.\] 
	Then $\tau_0$ appears as the rightmost pivot of $\macmatred_d$. Since 
	\[
		\{\sigma\in\degdmon:\sigma\succ\tau\}\subseteq\LM(I) 
	,\] 
	all columns to the left of that indexed by $\tau_0$ contain a pivot.

	Conversely, suppose $\macmatred_d$ takes the desired form. For any
	$\tau\in\LM(I)$ of degree $d$, the column indexed by $\tau$ contains a
	pivot, and thus belongs to the left identity block. Thus, any column to
	the left of that indexed by $\tau$ must also contain a pivot.
\end{proof}

We conclude by conjecturing that comaximal determinantal ideals of matrices of
linear forms are indeed reverse lexicographic, and describing the appropriate
Zariski open subset which must be nonempty in order for this to be true.

\begin{conjecture}\label{conj:revlex:det-ideal-is-revlex}
	Let $\RL$ be the property defined by
	 \[
		 \RL(I)=
		 \begin{cases}
			 \texttt{true}\quad &\text{if } I \text{ is reverse lexicographic}\\
			 \texttt{false} &\text{otherwise}
		 \end{cases}
	.\] 
	Then for any $n\ge 3$, $\RL$ is $(4,n-2,n,1)$-generic. 
\end{conjecture}
By \cref{lem:revlex:revlex-macaulay}, it is sufficient to prove that there
exists some Zariski open subset $U\subseteq\affspaceCorkOne$ such that for all
$a\in U$, for any $d\in\ZZpos$, the reduced Macaulay matrix
$\macmatred_d(\detsystemCorkOne[\phi_a({\genmat[1]})])$ takes the form
\[
\macmatred_d(\detsystemCorkOne[\phi_a({\genmat[1]})])=
\left(
	\begin{array}{c|c}
		I & X
	\end{array}
\right)
\] 
after a suitable row permutation and removal of reductions to zero. Recall that
the notation used here was introduced in \cref{sec:genericity}.

By \cref{prop:hilb:hilb}, there exists a Zariski open subset
$U_{\HS}\subseteq\affspaceCorkOne$ such that for all $a\in U_{\HS}$, the
Hilbert series of $\detidealCorkOne[\phi_a({\genmat[1]})]$ is the one given in
\cref{prop:hilb:hilb}. For any $d\in\ZZpos$, let
$h_d=\HF_{\detidealCorkOne[\phi_a({\genmat[1]})]}(d)$. 

Now fix some $d\in\ZZpos$, $n-1\le d\le 2n-3$. The determinant of the square
submatrix of the Macaulay matrix
$\macmatred_d(\detsystemCorkOne[{\genmat[1]}])$ given by the first $h_d$
columns after removing zero rows is a polynomial in $\fraka$,
$g_d(\fraka)\in\field[\fraka]$.

The distinguished Zariski open set $\affspaceCorkOne\smallsetminus
V(g_d(\fraka))$ consists precisely of those $a\in\affspaceCorkOne$ such that
\[
\macmatred_d(\detsystemCorkOne[{\genmat[1]}])=\left(\begin{array}{c|c}I
& X\end{array}\right)
.\] 
By \cite[Corollary~19]{FaugereSafeySpaenlehauer2013}, there exists a Zariski
open subset $O\subseteq\affspaceCorkOne$ such that the largest degree Macaulay
matrix which needs to be reduced is $\macmat_{2n-3}$. Thus, letting
\[
	U_{\RL}=O\cap U_{\HS}\cap\bigcap_{d=n-1}^{2n-3}\left(\affspaceCorkOne\smallsetminus V_{\algclosure}(g_d(\fraka))\right)
,\] 
for any $a\in U_{\RL}$, the ideal $\detidealCorkOne[\phi_a({\genmat[1]})]$ is
reverse lexicographic.

It is not clear, however, that the set $U_{\RL}$ is nonempty. Equivalently, it
is not clear that the polynomials  $g_d(\fraka)\in\field[\fraka]$ are nonzero.

In \cite[Theorem~3]{Pardue10} (see also the references therein), necessary and
sufficient conditions are given in order for a given power series to be the
Hilbert series of a reverse lexicographic ideal. By
\cref{lem:analysis:h-increasing}, these assumptions are satisfied by the
Hilbert series given in \cref{prop:hilb:hilb}.

In the following section, we relate \cref{conj:revlex:det-ideal-is-revlex} to
the Lefschetz properties (see \cite{Harima2013ch3}), and show that under
certain genericity assumptions, $\detidealCorkOne$ has the \emph{weak Lefschetz
property}.

%% file: lefschetz.tex
\section{Determinantal ideals and the Lefschetz properties}\label{sec:lefschetz}

We devote this section to exploring conditions under which determinantal ideals
possess the so-called Lefschetz properties, in the hope of shedding some light
on \cref{conj:revlex:det-ideal-is-revlex}. These properties have been widely
studied in various contexts, notably that of Artinian Gorenstein algebras.

This section is entirely self-contained, and for ease of exposition we do not
provide definitions of several classical properties from commutative algebra
(e.g. Artinian, Gorenstein). Such definitions and a wealth of related facts can
be found, e.g., in \cite{Matsumura1987}.

\begin{definition}[{\cite[Definitions\,3.1 and 3.8]{Harima2013ch3}}]
	Let $A=\bigoplus_{d=0}^{c}A_d$ be a graded Artinian $\field$-algebra
	with $A_c\ne 0$. The algebra $A$ has the \emph{weak Lefschetz
	property}, or simply WLP if there exists some $\ell\in A_1$ such that
	for all $0\le d\le c-1$, the map of $\field$-vector spaces
	\[
		\times\ell:A_{d}\to A_{d+1}
	\] 
	given by multiplication by $\ell$ has full rank. Such an $\ell$ is
	called a \emph{weak Lefschetz element}. If, in addition, for all $0\le
	d\le c-1$ and  $1\le s\le c-d$, the map
	\[
		 \times\ell^{s}:A_d\to A_{d+s}
	\] 
	has full rank, then $A$ is said to have the \emph{strong Lefschetz
	property}, or simply SLP, and  $\ell$ is called a \emph{strong
	Lefschetz element}.
\end{definition}

A natural generalization is the $t$-Lefschetz property.

\begin{definition}[{\cite[Definition\,6.1]{Harima2013ch6}}]
	Let $A=\bigoplus_{d=0}^{c}A_d$ be a graded Artinian $\field$-algebra.
	For some $t\ge 1$, $A$ has the $t$-WLP (resp. $t$-SLP) if there are
	some $\ell_1,\dots,\ell_t\in A_1$ such that $\ell_1$ is a weak (resp.
	strong) Lefschetz element for $A$ and, for each $1<i\le t$, the linear
	form $\ell_i$ is a weak (resp. strong) Lefschetz element for
	$A/(\ell_1,\dots,\ell_{i-1})$.
\end{definition}

When $A$ is an Artinian ideal of a polynomial ring over $\field$, the notion of
the $t$-SLP has a useful description in terms of the Hilbert series of $A$. 

\begin{proposition}[{\cite[Remark\,6.11]{Harima2013ch6}}]
	Let $I\subseteq\ring$ be a graded Artinian ideal. Then $\ell$ is a
	strong Lefschetz element for $\ring/I$ if and only if for all $s\ge 1$,
	\[
		\HF_{\ring/(I+\langle
		\ell^{s}\rangle)}(d)=\max\{\HF_{\ring/I}(d)-\HF_{\ring/I}(d-s),
	0\}
	.\] 
\end{proposition}

If $I\subseteq\ring$ is an Artinian ideal, then $\ring/I$ has dimension zero.
Thus, there is some $D$ such that for all $d>D$, $\HF_{\ring/I}(d)=0$. Clearly,
one can restrict to $s\le D$ in the above proposition.

\begin{theorem}[{\cite[Corollary\,6.30]{Harima2013ch6}}]
	Let $I\subseteq\field[x_1,x_2,x_3,x_4]$ be a graded Artinian ideal such
	that $\ring/I$ has the $2$-SLP. Then the generic initial ideal
	$\gin(I)$ of $I$ with respect to the grevlex order is the unique weakly
	reverse lexicographic ideal with Hilbert function $\HF_{\ring/I}$.
\end{theorem}

To establish \cref{conj:revlex:det-ideal-is-revlex}, it is therefore sufficient
to prove, in our setting, that the $2$-SLP is $(4,n-2,n,1)$-generic. On the
other hand, while it would not establish
\cref{conj:revlex:det-ideal-is-revlex}, it would still be useful to discern
whether or not the WLP is $(4,n-2,n,1)$-generic.

\begin{theorem}[{\cite[Remark\,4.4]{Migliore2003}}]\label{thm:lefschetz:det-ideal-wlp}
	Let $I\subseteq\field[x_1,\dots,x_k]$ be a graded ideal. If  $I$ is
	Artinian and has no generator in degree $1$, and if the quotient
	$\ring/I$ is Gorenstein and compressed with even socle degree, then
	$\ring/I$ has the weak Lefschetz property.
\end{theorem}

The socle degree of a zero-dimensional ideal $I\subseteq\field[x_1,\dots,x_k]$
is simply the degree of its Hilbert series, which is a polynomial. In this
context, to be compressed simply means that $\HF_{\ring/I}(d)=\HF_{\ring}(d)$
for all $0\le d\le \frac{s}{2}$, where $s=\deg(H_{\ring/I}(t))$ is the socle
degree of $\ring/I$.

By \cite[Theorem\,10]{FaugereSafeySpaenlehauer2013}, the property of being
zero-dimensional is $(4,n-2,n,1)$-generic, thus so is the property of being
Artinian. 

As soon as $n\ge 3$, the $2$-minors of $M$ are of degree at least $2$ and
generate $\detidealCorkOne$. Again by
\cite[Theorem\,10]{FaugereSafeySpaenlehauer2013}, the property that $\ring/I$
has socle degree exactly $2n-2$ is $(4,n-2,n,1)$-generic.

If $M$ is a matrix of linear forms the degree of the $(n-1)$-minors of $M$ is
$n-1$. Thus, the property that $\ring/I$ is compressed is also
$(4,n-2,n,1)$-generic.

Finally, by \cite[Theorem\,2.26]{BrunsVetter1988}, the Gulliksen-Neg{\aa}rd
complex is a free resolution of $M$ as soon as $\detidealCorkOne$ is
Cohen-Macaulay and the final module of this resolution has rank one (see
\cref{subsec:gn}). Thus, the property that $\detidealCorkOne$ is Gorenstein is
also  $(4,n-2,n,1)$-generic.

%
%

%% file: algorithm.tex
\section{A signature-based Gr\"obner basis algorithm for \texorpdfstring{$\detidealCorkOne$}{I[n-2](M)}}
We describe here the algorithm \textsc{DetGB}, an altered version of the $F_5$
algorithm from \cite{GoNeSa23}, which given a matrix $M$ of linear forms over
$\ring$ computes the reduced grevlex Gr\"obner basis for $\detidealCorkOne$. It
is precisely this algorithm which we analyze in the subsequent section. In
contrast to the standard matrix-$F_5$ algorithm (see \cite{Faugere2002} and
\cite{BardetFaugereSalvy2015}), the algorithm which we analyze does not compute
Gr\"obner bases for subsequences of the input sequence. However, as in the
matrix-$F_5$ algorithm, the algorithm described below does compute the
Gr\"obner basis degree by degree.

\subsection{A signature-based Gr\"obner basis algorithm for modules}
We begin by describing an algorithm (\cref{alg:modgb}) which, given a set
$\modF$ of homogeneous elements (all of the same degree) of the free module
$\ring^{m}$, a monomial order $\succ$ on $\ring$, a subset
$Z\subseteq\LMmod(\Syz(\modF))$, and a degree bound $D$, computes the reduced
$D$-$\modord$-Gr\"obner basis of $\langle\modF\rangle$ while avoiding those
reductions to zero which arise from the leading monomials given by $Z$.

\begin{algorithm}[ht]
	\caption{$\modgb(F,\succ,Z,D)$}
	\label{alg:modgb}
	\begin{algorithmic}[1]
	\Require{A collection of homogeneous module elements $\modF=\{\modf_1,\dots,\modf_s\}\subseteq\ring^{m}$ all of degree $d_0$, a monomial order $\succ$ on $\ring$, a subset $Z\subseteq\LMmod(\Syz(\modF))$, and a degree bound $D$.}
		\Ensure{The reduced $D$-Gr\"obner basis of the module $\langle \modF\rangle$ with respect to the module order $\mathrm{TOP}_\succ$.}
		\For{$i\in[1,\dots,s]$}
		\State $\macmat_{d_0}\gets$ concatenate $\modf_i$ to $\macmat_{d_0}$ with signature $(i,1)$
		\EndFor 
		\State $\macmatred_{d_0}\gets\rref(\macmat_{d_0})$ 
		\State $G\gets\rows(\macmatred_{d_0})$
		\For{$d\in[d_0+1,\dots, D]$}
		\For{$g\in\rows(\macmatred_{d-1})$}
		\State $(i,\tau)\gets\signature(g)$
		\For{$j\in\max_{k}\{x_k\mid\tau\}$}
		\If{$x_k\tau e_i\notin Z$}
		\State $\macmat_d\gets$ concatenate $x_kg$ to  $\macmat_d$ with signature $(i,x_k\tau)$
		\EndIf
		\EndFor
		\EndFor
		\State $\macmatred_{d}\gets\rref(\macmatred_d)$
		\State $G\gets G\cup\rows(\macmat_d)$
		\EndFor
		\State \Return $G$
	\end{algorithmic}
\end{algorithm}

In \cref{alg:modgb}, $\rref(\macmat_d)$ is the reduced row echelon form of a
Macaulay matrix $\macmat_d$, and $\rows(\macmatred_d)$ is the set of rows of
$\macmatred_d$, interpreted as elements of $\ring^{m}$.

The algorithm works by building Macaulay matrices in various degrees for the
module $\langle \modf_1,\dots,\modf_s\rangle$. It then echelonizes these
Macaulay matrices using a general-purpose echelon form algorithm. From these
echelonized matrices, it extracts polynomials whose leading terms do not belong
to the ideal generated by the leading terms of the intermediate Gr\"obner basis
computed. \cref{alg:modgb} uses the additional data of leading
terms of syzygies (the input $Z$) to avoid adding to the Macaulay matrix rows which are known
to reduce to zero upon
echelonization. Furthermore, it builds the Macaulay matrix
$\macmat_d$ from $\macmatred_{d-1}$ rather than from the original system
$\modf_1,\dots,\modf_s$ since in doing so, a portion of $\macmat_d$ will
already be echelonized. It is precisely by exploiting this specific structure
of $\macmat_d$ that we arrive at sharp complexity analyses in \cref{sec:analysis}.

The termination and correctness of \cref{alg:modgb} follow essentially from
\cref{thm:preliminaries:macmat-echelonization}. A detailed proof can be found
in \cite[Theorem~9]{BardetFaugereSalvy2015}.

\begin{remark}
	When $m=1$, \cref{alg:modgb} is essentially just Lazard's algorithm
	(see \cite{Lazard1983}), with the additional input of a set of
	precomputed syzygies. Given this, the standard matrix-$F_5$ algorithm
	is recovered as a very slight alteration to \cref{alg:modgb} by
	updating the set $Z$ with the leading monomials of the matrices
	$\macmatred_d$ along the way. 
\end{remark}

\subsection{The \textsc{DetGB} algorithm}

Using \cref{alg:modgb} combined with syzygy information from the
Gulliksen-Neg{\aa}rd complex leads to \cref{alg:detgb} (see also
\cite[Algorithm~3]{GoNeSa23}) to compute Gr\"obner bases for the determinantal
ideals considered in this paper.

\begin{algorithm}
	\caption{$\detgb(M)$}
	\label{alg:detgb}
	\begin{algorithmic}[1]
		\Require{An $n\times n$ matrix $M$ of homogeneous linear forms in four variables.}
		\Ensure{The reduced grevlex Gr\"obner basis for $\detidealCorkOne$.}
		\State $S_2\gets$ a set of generators for $\Syz_2(\detsystemCorkOne)$ computed using the Gulliksen-Neg{\aa}rd complex.
		\State $S_1\gets$ a set of generators for $\Syz(\detsystemCorkOne)$ computed using the Gulliksen-Neg{\aa}rd complex.
		\State $L_2\gets\modgb(S_2,\mathrm{grevlex},\emptyset,n-3)$
		\State $L_1\gets\modgb(S_1,\mathrm{grevlex},\LMmod(L_2),n-2)$
		\State \Return $\modgb(\detsystemCorkOne,\mathrm{grevlex},\LMmod(L_1),2n-3)$
	\end{algorithmic}
\end{algorithm}

The termination and correctness of \cref{alg:detgb} is proven in
\cite[Proposition~21]{GoNeSa23}.

\begin{remark}
	If $\detidealCorkOne$ is not Cohen-Macaulay, the Gulliksen-Neg{\aa}rd
	complex need not be a free resolution for $\detidealCorkOne$. However,
	it is still a complex. Thus, the syzygy modules of $\detidealCorkOne$
	contain, possibly properly, the modules computed from the
	Gulliksen-Neg{\aa}rd complex. It is for this reason that we do not need
	to assume any genericity properties in order for \cref{alg:detgb} to be
	correct.

	See \cite[Remarks~10 and~20]{GoNeSa23} for a
	more detailed discussion.
\end{remark}

\begin{remark}
	The image of the standard basis elements of $\freeresmod_0$ under the
	augmentation map $\epsilon$ in the Gulliksen-Neg{\aa}rd complex are
	actually the cofactors of order $n-1$ of $M$, not the minors of order
	$n-1$. Computing the images of the boundary maps in the
	Gulliksen-Neg{\aa}rd complex therefore gives syzygy modules for the
	cofactors of order $n-1$, rather than the minors, as we would like.
	This can be corrected by simply replacing $\detsystemCorkOne$ with the
	cofactors of order $n-1$ of $M$ in \cref{alg:detgb}. Alternatively, we
	can easily turn the syzygies of the cofactors obtained from the
	Gulliksen-Neg{\aa}rd complex into syzygies of the minors, as explained
	in the proof of \cite[Theorem~9]{GoNeSa23}.
\end{remark}

%% file: analysis.tex
\section{Complexity analysis}
\label{sec:analysis}

We consider a matrix of the form $M=\phi_a(\genmat[1])$, where $a$ is taken to
be a point in some suitable Zariski open subset of $\affspaceCorkOne$. We make
precise which Zariski open subsets we must take $a$ to lie in below, appealing
to the various genericity statements we have established above (e.g.
\cref{prop:genericity:det-ideal-is-cm,prop:hilb:hilb,conj:revlex:det-ideal-is-revlex}).
Our complexity analysis begins by computing the number of polynomials of each
degree in the reduced grevlex Gr\"obner basis of $\detidealCorkOne$. 

The coefficient of $t^d$ in the Hilbert series of $\detidealCorkOne$ is, by
definition, the dimension of the $\field$-vector space of homogeneous
polynomials in $\detidealCorkOne$ of degree $d$. This dimension is also
precisely the rank of the Macaulay matrix of $\detsystemCorkOne$ in degree $d$.
Combining these ranks with the aforementioned count of polynomials of degree $d$
in the reduced grevlex Gr\"obner basis of $\detidealCorkOne$ allows us to
compute tight bounds on the complexity of the overall Gr\"obner basis
computation using fast linear algebra techniques.

Following the standard for complexity bounds, we use the Bachmann-Landau
notation $O(\cdot)$ (see e.g.\ \cite[Section~3.1]{Cormen22}).

\subsection{Bounding \texorpdfstring{$\#(\LM(G_{d+1})\smallsetminus\LM(G_{d}))$}{the growth of the Gr\"obner basis}}

The work of computing the number of polynomials of degree $d$ in the reduced
grevlex Gr\"obner basis for $\detidealCorkOne$ is already accomplished by our
analysis of staircases of reverse lexicographic ideals in \cref{sec:revlex_ideals}. The following
proposition arises from plugging in the relevant quantities into the formula
given in \cref{prop:revlex:new-gb-bound}.

\begin{proposition}\label{prop:analysis:new-gb-size}
	Suppose \cref{conj:revlex:det-ideal-is-revlex} is true. Fix $a\in
	U_{\RL}\cap U_{\HS}\subseteq\affspaceCorkOne$. Let
	$M=\phi_a({\genmat[1]})$. For any integer $n-1\le d<2n-3$, 
	\[
		\#(\LM(G_{d+1})\smallsetminus\LM(G_d))=\frac{(d-2n+3)(d-2n+2)}{2}
	,\]
	where $G_d$ is the reduced $d$-Gr\"obner basis for $\detidealCorkOne$
	with respect to the grevlex order.
\end{proposition}
\begin{proof}
	If \cref{conj:revlex:det-ideal-is-revlex} is true, then the ideal
	$\detidealCorkOne$ is reverse lexicographic. Therefore, we can apply
	\cref{prop:revlex:new-gb-bound}. We begin by showing that for all $d\ge
	n-1$, the integer $\ell_d$ is $3$. Recall, from the statement of
	\cref{prop:revlex:new-gb-bound}, that $\ell_d$ is defined to be the
	largest integer such that 
	\[
		\binom{\ell_d+d-1}{\ell_d-1}<h_d
	,\] 
	where $h_d=\HF_{\detidealCorkOne}(d)$. For any $a\in U_{\HS}$, the
	Hilbert series $H_{\detidealCorkOne}(t)$ is given by
	\cref{prop:hilb:hilb}.

	First, note that $h_{n-1}=n^2$ and
	\[
		\binom{(n-1)+2}{2}=\frac{n^2+n}{2}
	.\]
	Since $n>1$, this shows that $\ell_{n-1}\ge 3$.	As $k=4$,
	$\ell_{n-1}\le 3$.

	We proceed by induction. Suppose $\ell_{d}=3$ for some $d\ge n-1$. Then
	there must be at least one monomial $\tau$ in which the variable $x_4$
	appears in $\LM(G_d)$. Subsequently for any variable $x$,
	$x\tau\in\LM(G_{d+1})$ and $x\tau$ contains the variable $x_4$ as well,
	showing that $\ell_{d+1}=3$. 

	Finally, applying \cref{prop:revlex:new-gb-bound}, we obtain
	\begin{align*}
		\#\LM(G_{d+1})\smallsetminus\LM(G_d)&=h_{d+1}-h_d+\binom{d}{1}+2\binom{d+1}{2}-3\binom{d+2}{2}\\
						    &=h_{d+1}-h_d+d+2\binom{d+1}{2}-3\binom{d+2}{2}\\
						    &=\frac{(d-2n+3)(d-2n+2)}{2}
	.
  \qedhere
  \end{align*} 
\end{proof}

\subsection{Lower bounds}

As a first application of \cref{prop:analysis:new-gb-size}, we establish an
exact expression for the size of the reduced grevlex Gr\"obner basis of ideals
of the form $\detidealCorkOne$ under certain genericity assumptions.

\begin{proposition}\label{prop:analysis:total-gb}
	Fix $a\in U_{\RL}\subseteq\affspaceCorkOne$. Let
	$M=\phi_a({\genmat[1]})$. The total number of polynomials in the
	reduced grevlex Gr\"obner basis for $\detidealCorkOne$ is
	\[
		\#G=\frac{n(n+1)(n+2)}{6}
	.\] 
\end{proposition}
\begin{proof}
	Enumerating the polynomials in the reduced grevlex Gr\"obner basis for
	$\detidealCorkOne$ is equivalent to enumerating the leading monomials
	of these polynomials. That is,
	\[
		\#G=\#\LM(G_{n-1})+\sum_{d=n-1}^{2n-4}\#(\LM(G_{d+1})\smallsetminus\LM(G_d))
	.\] 
	Using \cref{prop:analysis:new-gb-size},
	\[
		\#G=n^2+\sum_{d=n-1}^{2n-4}\frac{(d-2n+3)(d-2n+2)}{2}=\frac{n(n^2+3n+2)}{6}
	.
  \qedhere
  \] 
\end{proof}

Recall that here, we consider the computation of a dense representation of the
sought Gr\"obner basis, meaning that all coefficients in \(\field\) of all
elements in this basis are explicitly computed.

The expression obtained in \cref{prop:analysis:total-gb} counts only the
leading monomials of the polynomials in the reduced grevlex Gr\"obner basis of
$\detidealCorkOne$, and not the smaller monomials in these polynomials. In the
following theorem, we compute \,---\,under our genericity assumptions\,---\,the
number of nonzero coefficients in each of these polynomials.

\begin{theorem}\label{thm:analysis:lower-bound}
	Suppose \cref{conj:revlex:det-ideal-is-revlex} is true. Fix $a\in
	U_{\RL}\cap U_{\HS}\subseteq\affspaceCorkOne$. Let
	$M=\phi_a({\genmat[1]})$. The number of elements of $\field$ in the
	dense representation of the reduced grevlex Gr\"obner basis of
	$\detidealCorkOne$ is asymptotically bounded from below by $n^{6}$.
\end{theorem}
\begin{proof}
	Since $a\in U_{\HS}$, the Hilbert series of $\detidealCorkOne$ is the
	one given in \cref{prop:hilb:hilb}. Let
	$h_d=\HF_{\ring/\detidealCorkOne}(d)$. The number of monomials
	appearing in a degree $d$ polynomial in the reduced grevlex Gr\"obner
	basis for $\detidealCorkOne$ is
	\[
		\binom{3+d}{3}-h_d+1
	.\] 
	Therefore, using \cref{prop:analysis:new-gb-size}, we find that the
	number of nonzero monomials appearing in the grevlex Gr\"obner basis for
	$\detidealCorkOne$ is
	\begin{align*}
		N=n^2\bigg(\binom{2+n}{3}&-n^2+1\bigg)\\
					&+\sum_{d=n-1}^{2n-4}\frac{(d-2n+3)(d-2n+2)}{2}\left(\binom{3+d}{3}-h_d+1\right) 
	.\end{align*} 
	Expanding this gives 
	\[
	N=\frac{1}{72}n^6+\frac{13}{120}n^5-\frac{4}{9}n^4+\frac{13}{24}n^3+\frac{31}{72}n^2+\frac{7}{20}n\qedhere
	\] 
\end{proof}

\subsection{Upper bounds}

For a given degree $d>n-1$, \cref{alg:detgb} builds the Macaulay matrix in
degree $d$ by multiplying each row of the Macaulay matrix in degree $d-1$ by
each variable, utilizing the signatures attached to each row to avoid
redundancies in rows. The reverse lexicographic property of determinantal
ideals provides the unreduced Macaulay matrix in degree $d$ with a precise
structure, which we analyze to obtain complexity upper bounds.

\begin{proposition}\label{prop:analysis:macmat-structure}
	Suppose \cref{conj:revlex:det-ideal-is-revlex} is true. Fix $a\in
	U_{\RL}\subseteq\affspaceCorkOne$. Let $M=\phi_a({\genmat[1]})$ and let
	$d\in\ZZpos$, $n-1\le d<2n-3$. Then after a suitable row permutation,
	the unreduced Macaulay matrix $\macmat_{d+1}$ of $\detidealCorkOne$
	built by \cref{alg:detgb} is of the form
	\[
	\left(
		\begin{array}{c|c}
			T_{d+1} & X_{d+1}\\ \hline
			A_{d+1} & Y_{d+1}
		\end{array}
	\right)
	\] 
	where $T_{d+1}$ is a square upper triangular block of size
	$h_{d+1}-\frac{(d-2n+3)(d-2n+2)}{2}$ and
	$h_{d+1}=\HF_{\detidealCorkOne}(d+1)$ is the Hilbert function of
	$\detidealCorkOne$ evaluated at $d+1$.
\end{proposition}
\begin{proof}
	Let $G$ be the reduced grevlex Gr\"obner basis of $\detidealCorkOne$.
	We partition the rows of $\macmatred_{d+1}$ into the following two sets
	\[
		R_1=\{f\in\rows(\macmatred_{d+1}):f\in G\}\quad R_2=\{f\in\rows(\macmatred_{d+1}):f\notin  G\}
	.\] 
	Let $\tau=\min_{f\in R_2}\{\LM(f)\}$. Then there must exist some
	variable $x_j$ such that $\frac{\tau}{x_j}\in\LM(G)$. Fix $\sigma\in
	R_1$. In \cref{prop:analysis:new-gb-size} it was shown that all
	monomials involving $x_1,x_2,x_3$ appear in $\LM(G_{n-1})$. Thus, since
	$d+1>n-1$, $x_4\mid\sigma$. As $G$ is a reduced Gr\"obner basis and
	$\sigma\in\LM(G)$, the monomial $\frac{\sigma}{x_4}$ is not in
	$\LM(G)$. Since we assume \cref{conj:revlex:det-ideal-is-revlex}, this
	forces $\frac{\tau}{x_j}\succ\frac{\sigma}{x_4}$. Subsequently, since
	$x_j\succ x_4$, we have that $\tau\succ\sigma$.

	This shows that any monomial of $\LM(R_1)$ is smaller than $\tau$. On
	the other hand, assuming \cref{conj:revlex:det-ideal-is-revlex}, for
	any $\tau'\succ\tau$, there exists some $g\in R_2$ such that
	$\LM(g)=\tau'$. 

	Now for any polynomial $g\in R_2$, there exists some
	$h\in\rows(\macmatred_d)$ such that $\LM(h)\mid\LM(g)$. Since
	\cref{alg:modgb} constructs the rows of $\macmat_{d+1}$ by multiplying
	the rows of $\macmatred_d$ by suitable variables, we see that there
	must be a row of $\macmat_{d+1}$ with leading monomial precisely
	$\LM(g)$. Thus, the set of rows of $\macmat_{d+1}$ which, upon
	echelonization, are in $R_2$ form a submatrix of $\macmat_{d+1}$ of the
	form
	\[
	\left(
		\begin{array}{c|c}
			T_{d+1} & X_{d+1}
		\end{array}
	\right)
	\] 
	with $T_{d+1}$ upper triangular. The rows of $T_{d+1}$ are in bijection
	with $R_2$, and by \cref{prop:analysis:new-gb-size}, the set $R_1$ has
	cardinality $\frac{(d-2n+3)(d-2n+2)}{2}$. As $\macmat_{d+1}$ has
	exactly $h_{d+1}$ rows, we see that $T_{d+1}$ has
	$h_{d+1}-\frac{(d-2n+3)(d-2n+2)}{2}$ rows.
\end{proof}

By \cref{prop:analysis:macmat-structure}, $A_{d+1}$ is a matrix with
$\alpha_{d+1}=\frac{(d-2n+3)(d-2n+2)}{2}$ rows, $T_{d+1}$ is a matrix with
$\beta_{d+1}=h_{d+1}-\alpha_{d+1}$ rows, and and $X_{d+1}$ is a matrix with
$\gamma_{d+1}=\binom{4+d}{3}-\beta_{d+1}$ columns. We begin by establishing
various useful facts about the behavior of $h_{d+1},
\alpha_{d+1},\beta_{d+1},\gamma_{d+1}$.

\begin{lemma}\label{lem:analysis:alpha-decreasing}
	For all $n-1<d<2n-3$, $\alpha_{d+1}<\alpha_d$.
\end{lemma}
\begin{proof}
	Let $f(d)=\alpha_d$ and view $f$ as a continuous function in one real
	variable. Then $f'(d)=d-2n+\frac{5}{2}$ has its unique root at
	$d=2n-\frac{5}{2}$. Thus, $f$ is strictly decreasing on the interval
	$\left(-\infty,2n-\frac{5}{2}\right]$.
\end{proof}

\begin{lemma}\label{lem:analysis:alpha-beta}
	For all $n-1\le d<2n-3$, $\alpha_{d+1}<\beta_{d+1}$.
\end{lemma}
\begin{proof}
	In view of
	\cref{lem:analysis:h-increasing,lem:analysis:alpha-decreasing}, it
	suffices to show that $\alpha_n<\frac{h_n}{2}$. We have
	$\alpha_n=\frac{n^2-3n+2}{2}$ and $h_n=2n^2+2$. Thus,
	$\frac{h_n}{2}-\alpha_n=\frac{n^2+3n}{2}$, which is certainly positive
	for $n\ge 3$.
\end{proof}

\begin{lemma}\label{lem:analysis:alpha-gamma}
	For all $n-1\le d<2n-3$, $\alpha_{d+1}\le\gamma_{d+1}$.
\end{lemma}
\begin{proof}
	Recall that $h_{d+1}=\dim_{\field}(\detidealCorkOne_{d+1})$. The
	$\field$-vector space $\detidealCorkOne_{d+1}$ is a sub-$\field$ vector
	space of $\ring_{d+1}$. Since
	$\dim_{\field}(\ring_{d+1})=\binom{4+d}{3}$, we have that
	$h_{d+1}\le\binom{4+d}{3}$. Finally,
	\begin{align*}
		\gamma_{d+1}&=\binom{4+d}{3}-\beta_{d+1}\\
			    &=\binom{4+d}{3}-h_{d+1}+\alpha_{d+1}\\
			    &\ge\alpha_{d+1}
	\end{align*}
\end{proof}

Before turning to the complexity of echelonizing a Macaulay matrix in a fixed
degree, we need one final auxiliary lemma regarding the cost of solving
several triangular systems.

\begin{lemma}[see {\cite[Lemma~3.1]{DuGiPe04}}]\label{lem:analysis:trsm}
	Let $U\in\field^{p\times p}$ be an invertible $p\times p$ upper
	triangular matrix and $V\in\field^{p\times q}$ a $p\times q$ matrix.
	Let $C(p,q)$ be the arithmetic complexity of computing $U^{-1}V$ using
	\cite[\texttt{ULeft-TRSM}$(U,V)$]{DuGiPe04}. Then 
	\[
		C(p,q)\in\begin{cases}
			O(qp^{\omega-1})&\text{if }p\le q\\
			O(p^2q^{\omega-2})&\text{if }p>q
		\end{cases}
	.\] 
\end{lemma}
\begin{proof}
	For $p\le q$, this is precisely the statement of
	\cite[Lemma~3.1]{DuGiPe04}. Assume then that $p>q$. In the following,
	we denote by $C_\omega$ the constant associated to rectangular matrix
	multiplication with exponent $\omega$. That is, the cost of multiplying
	a $p\times q$ matrix by a $q\times s$ matrix (all with entries in
	$\field$) is bounded by $C_\omega
	\min\{p,q,s\}^{\omega-2}\max\{pq,ps,qs\}$. We use the case $p\le q$ as
	a base case for the recursion. The complexity in this case is given
	directly by \cite[Lemma~3.1]{DuGiPe04}. That is,
	\[
		C(p,q)=\begin{cases}
			\frac{C_\omega}{2(2^{\omega-2}-1)}qp^{\omega-1}&\text{if }p\le q\\
			C(\left\lceil \frac{p}{2} \right\rceil,q)+C(\left\lfloor \frac{p}{2} \right\rfloor,q)+C_\omega p^2q^{\omega-2} &\text{otherwise}
		\end{cases}
	.\] 
	By padding $U$ with an identity block and $V$ with zeroes, we may
	assume that $\frac{p}{q}$ is a power of two. Subsequently, we have
	\begin{align*}
		C(p,q)&=\frac{p}{q}C(q,q)+C_\omega p^2q^{\omega-2}\sum_{j=0}^{\log_2\left(\frac{p}{q}\right)-1}\frac{1}{2^{j}}\\
		      &=\frac{C_\omega}{2(2^{\omega-2}-1)}pq^{\omega-1}+2C_\omega p^2q^{\omega-2}\left(1-\frac{q}{p}\right)\\
		      &=2C_\omega p^2q^{\omega-2}+\left(\frac{C_\omega}{2(2^{\omega-2}-1)}-2C_\omega\right)pq^{\omega-1}
	\end{align*} 
	as desired.
\end{proof}

Putting together
\cref{lem:analysis:h-increasing,lem:analysis:alpha-decreasing,lem:analysis:alpha-beta,lem:analysis:alpha-gamma},
we can compute an upper bound on the cost of echelonizing a Macaulay matrix in
a fixed degree.

\begin{proposition}\label{prop:analysis:reduction-complexity}
	For any $n-1\le d\le 2n-3$, the number of arithmetic operations in
	$\field$ required to compute the matrix $\macmatred_{d+1}$ from
	$\macmat_{d+1}$ is in 
	\[
		O\left(\beta_{d+1}^2\alpha_{d+1}^{\omega-2}+\alpha_{d+1}^{\omega-2}\beta_{d+1}\gamma_{d+1}\right)
	.\] 
\end{proposition}
\begin{proof}
	The computation of $\macmatred_{d+1}$ from  $\macmat_{d+1}$ can be
	broken up into four steps.

	\emph{Step 1.} First, we echelonize the upper block, which is of the
	form $\left(\begin{array}{c|c}T_{d+1}&X_{d+1}\end{array}\right)$, with
	$T_{d+1}$ upper triangular. Applying \cref{lem:analysis:trsm}, the cost
	of this step is $O(\beta_{d+1}^2\alpha_{d+1}^{\omega-2})$ 

	\emph{Step 2.} Next, we use the $I_{d+1}$ block to eliminate
	$A_{d+1}$. The resulting matrix takes the form
	\[
		\left(
			\begin{array}{c|c}
				I_{d+1} & X_{d+1}\\\hline
				0 & Y_{d+1}-A_{d+1}X_{d+1}
			\end{array}
		\right)
	.\] 
	The arithmetic complexity of this step is bounded by the cost of
	computing $A_{d+1}X_{d+1}$. The matrix $A_{d+1}$ has $\alpha_{d+1}$
	rows and $\beta_{d+1}$ columns, while $X_{d+1}$ has $\beta_{d+1}$ rows
	and $\gamma_{d+1}$ columns. By
	\cref{lem:analysis:alpha-beta,lem:analysis:alpha-gamma},
	\[
	\min\{\alpha_{d+1},\beta_{d+1},\gamma_{d+1}\}=\alpha_{d+1}
	.\] 
	Hence, by \cite[Section~2.1]{Knight1995}, the matrix $A_{d+1}X_{d+1}$ can
	be computed using $O(\alpha_{d+1}^{\omega-2}\beta_{d+1}\gamma_{d+1})$
	arithmetic operations in $\field$.
	
	\emph{Step 3.} Next, we compute the reduced row echelon form of
	$Y_{d+1}-A_{d+1}X_{d+1}$ which has $\alpha_{d+1}$ rows and
	$\beta_{d+1}$ columns. By \cref{lem:analysis:alpha-beta}, and using the
	general results of \cite[Section~2.2]{Storjohann2000} (see also
	\cite[Appendix~A]{JePeSt13}), this can be done using
	$O(\alpha_{d+1}^{\omega-1}\gamma_{d+1})$ operations in $\field$.

	\emph{Step 4.} The matrix after the previous step takes the form
	\[
		\left(
			\begin{array}{c|c|c}
				I_{d+1} & X_{d+1}^{(1)} & X_{d+1}^{(2)}\\ \hline
				0 & I^{(\alpha_{d+1})} & Y_{d+1}
			\end{array}
		\right)
	,\]
	where $X_{d+1}=\left(\begin{array}{c|c} X_{d+1}^{(1)} &
	X_{d+1}^{(2)}\end{array}\right)$ and $I^{(\alpha_{d+1})}$ is an
	identity matrix of size $\alpha_{d+1}$. The final step of the
	echelonization process is then to reduce $X_{d+1}$ using the identity
	block $I^{(\alpha_{d+1})}$. The resulting matrix takes the form
	\[
		\left(
			\begin{array}{c|c|c}
				I_{d+1} & 0 & X_{d+1}^{(2)}-X_{d+1}^{(1)}Y_{d+1} \\ \hline
				0 & I^{(\alpha_{d+1})} & Y_{d+1}
			\end{array}
		\right)
	.\] 
	Similarly to above, the arithmetic complexity of this step is bounded
	by that of computing $X_{d+1}^{(1)}Y_{d+1}$. The matrix $X_{d+1}^{(1)}$
	has $\beta_{d+1}$ rows and  $\alpha_{d+1}$ columns, while the matrix
	$Y_{d+1}$ has $\alpha_{d+1}$ rows and $\binom{4+d}{3}-h_{d+1}$ columns.
	Therefore, by the general bound given in
	\cite[Section~2.1]{Knight1995}, the number of arithmetic $\field$
	operations required to compute the matrix $X_{d+1}^{(1)}Y_{d+1}$ is in
	\[
	\begin{cases}
		O\left(\left(\binom{4+d}{3}-h_{d+1}\right)^{\omega-2}\alpha_{d+1}\beta_{d+1}\right) &\text{if } \alpha_{d+1}>\binom{4+d}{3}-h_{d+1}\\
		O\left(\left(\binom{4+d}{3}-h_{d+1}\right)\alpha_{d+1}^{\omega-2}\beta_{d+1}\right) &\text{otherwise} 
	\end{cases}
	\] 
	In the first case, 
	\[
		O\left(\left(\binom{4+d}{3}-h_{d+1}\right)^{\omega-2}\alpha_{d+1}\beta_{d+1}\right)\subseteq O(\alpha_{d+1}^{\omega-1}\beta_{d+1})\subseteq O(\alpha_{d+1}^{\omega-2}\beta_{d+1}\gamma_{d+1})
	\] 
	and in the second case, since $\binom{4+d}{3}-h_{d+1}\ge\gamma_{d+1}$,
	 \[
		 O\left(\left(\binom{4+d}{3}-h_{d+1}\right)\alpha_{d+1}^{\omega-2}\beta_{d+1}\right)\subseteq O(\alpha_{d+1}^{\omega-2}\beta_{d+1}\gamma_{d+1})
	\] 
	so the complexity of the second step dominates.
\end{proof}

Our main complexity upper bound, given in the following theorem, is now an easy consequence of \cref{prop:analysis:reduction-complexity}. 

\begin{theorem}\label{thm:analysis:final-complexity}
	Fix $a\in U_{\RL}\subseteq\affspaceCorkOne$. Let
	$M=\phi_a({\genmat[1]})$. The number of arithmetic operations in
	$\field$ performed by \cref{alg:detgb} when computing the reduced
	grevlex Gr\"obner basis for $\detidealCorkOne$ is in
	$O\left(n^{2\omega+3}\right)$.
\end{theorem}
\begin{proof}
	Note first that all arithmetic operations occur when computing the
	$\macmatred_d$ from the $\macmat_d$. Secondly, note that the complexity
	of the final step of \cref{alg:detgb} bounds the complexity of the
	algorithm as a whole, since the number of rows to be reduced in the
	Macaulay matrix in degree $d$ for the first (resp.\ second) syzygy
	module is precisely the number of (a priori) reductions to zero
	encountered in $\macmat_{d+1}$ (resp.\ the Macaulay matrix in degree
	$d+1$ of the first syzygy module). Note also that $\macmat_{n-1}$ has
	$n^2$ rows, and $\binom{n+2}{3}$ columns, and is of rank $n^2$. Thus,
	by \cite[Section~2.2]{Storjohann2000} (see also
	\cite[Appendix~A]{JePeSt13}) the arithmetic complexity of computing
	$\macmatred_{n-1}$ from $\macmat_{n-1}$ is in
	\[
		O\left(n^{2\omega-2}\binom{n+2}{3}\right)
    \;\;\subseteq\;\;
    O\left(n^{2\omega+1}\right)
	.\] 
	It follows, by \cref{prop:analysis:reduction-complexity}, that the
	total complexity of computing the reduced grevlex Gr\"obner basis for
	$\detidealCorkOne$ is in
  \(O(n^{2\omega+1}+f_\omega(n))\),
	where
	\[
		f_\omega(n)=\sum_{d=n-1}^{2n-4}\beta_{d+1}^2\alpha_{d+1}^{\omega-2}+\alpha_{d+1}^{\omega-2}\beta_{d+1}\gamma_{d+1}
	.\] 
	By \cref{lem:analysis:alpha-decreasing}, for all $n-1\le d\le 2n-4$,
	$\alpha_{d+1}\le\alpha_{n}<n^2$, hence
	\begin{equation}\label{eq:analysis:comp-2}
		f_\omega(n) \;\in\; O\left(n^{2\omega-4}\sum_{d=n-1}^{2n-4}\beta_{d+1}\gamma_{d+1}+\beta_{d+1}^2\right)
	.\end{equation}
	One can verify (e.g.\ using the Maple computer algebra system \cite{Maple}) that
	\[
		\sum_{d=n-1}^{2n-4}\beta_{d+1}\gamma_{d+1}+\beta_{d+1}^2=\frac{619}{1260}n^{7}-\frac{341}{360}n^{6}-\frac{7}{360}n^{5}+\frac{7}{36}n^{4}-\frac{169}{360}n^{3}-\frac{89}{360}n^{2}-\frac{1}{420}n
	.\] 
  It follows that \(f_\omega(n) \in O(n^{2\omega+3})\), which concludes the proof.
\end{proof}

\begin{remark}
	The upper bound on \cref{alg:detgb} obtained in
	\cref{thm:analysis:final-complexity} is subquadratic in the size of the
	dense representation of the output Gr\"obner basis obtained in
	\cref{thm:analysis:lower-bound}.
\end{remark}

\begin{remark}
	When $\omega=2$, the bound $O(n^{2\omega+3})$ becomes $O(n^{7})$, which
	still differs from the lower bound on the output size obtained in
	\cref{thm:analysis:lower-bound} by a factor of $n$. This suggests that
	there might still be room for improvement upon the bound obtained in
	\cref{thm:analysis:final-complexity}. In experiments, when working on
	input of matrices of homogeneous linear forms in four variables with
	coefficients chosen uniformly at random from some large prime field,
	one can observe that the submatrices $A_d$ defined in
	\cref{prop:analysis:macmat-structure} are sparse. Perhaps by taking
	into account this sparsity, a tighter upper bound could be achieved.
\end{remark}

\subsection{The asymptotic behavior of \texorpdfstring{$f_\omega(n)$}{the dominating function}}

In the proof of \cref{thm:analysis:final-complexity}, only one upper bound is
actually used\,---\,the bound $\alpha_{d+1}<n^2$. We conclude our complexity
analysis by presenting precise asymptotics for $f_\omega(n)$ for various
$\omega$. Using the SageMath computer algebra system (see \cite{sagemath}), we obtain
the asymptotic data in \cref{table:analysis:asymptotics}. It suggests that the
asymptotic result $f_\omega(n)=O\left(n^{2\omega+3}\right)$ obtained in
\cref{thm:analysis:final-complexity} is sharp.

\begin{table}[htp]
	\caption{The asymptotics of $f_\omega(n)$ compared to $n^{2\omega+3}$
	for various $2\le \omega\le 3$.}
	\label{table:analysis:asymptotics}
	\centering
	\begin{tabular}{|c|c|c|}
		\hline
		$\omega$ & $f_\omega(n)\sim_{n\to\infty}$ & $n^{2\omega+3}$\\\hline
		$3$ & $\frac{401}{18144}n^{9}$ & $n^{9}$\\
		$2.7$ & $\frac{2^{\frac{7}{10}}\cdot76533282553747476335323}{2761171875000000000000000}n^{8.4}$ & $n^{8.4}$\\
		$2.5$ & $\frac{29\sqrt{2}}{10080}n^{8}$ & $n^{8}$\\
		$2.38$ & $\frac{2^{\frac{19}{50}}\cdot3808710545424609640564981876343720387}{41658431291580200195312500000000000000}n^{7.76}$ & $n^{7.76}$\\
		$2$ & $\frac{619}{1260}n^{7}$ & $n^{7}$\\\hline
	\end{tabular}
\end{table}

%% file: acknowledgements.tex
\section*{Acknowledgments}

Funding: The author is supported by \emph{Quantum Information Center Sorbonne}
(QICS); by the Agence nationale de la recherche (ANR) [ANR-19-CE40-0018
\textsc{De Rerum Natura}, ANR-18-CE33-0011 \textsc{SESAME}, and
ANR-23-CE48-0003 \textsc{CREAM}]; the joint ANR-Austrian Science Fund FWF
[ANR-22-CE91-0007 \textsc{EAGLES} and ANR-FWF ANR-19-CE48-0015 \textsc{ECARP}];
and the EOARD-AFOSR [FA8665-20-1-7029].

Some computations required for the proof of
\cref{thm:analysis:final-complexity} were performed using
Maple${}^\mathrm{TM}$.

The author would like to thank his Ph.D.\ advisors Vincent Neiger and Mohab
Safey El Din for numerous helpful discussions and suggestions.

Finally, the author would like to thank Alessio Caminata for his explanation of
why the determinantal ideals considered here are indeed generically Gorenstein
(see the comments following \cref{thm:lefschetz:det-ideal-wlp}).